\documentclass[onefignum,onetabnum]{siamonline190516}

\usepackage{graphicx}
\usepackage{epstopdf}
\usepackage{dsfont}
\usepackage{amssymb, amsmath,amsfonts}
\usepackage{bm}
\usepackage{bbm}
\usepackage{tikz}
\usepackage{caption}
\usepackage{subcaption}
\usepackage{wrapfig}
\usepackage{tikz}
\usepackage{multirow}
\usetikzlibrary{shapes.geometric, arrows, positioning}
\usepackage{graphicx}% http://ctan.org/pkg/graphicx
\usepackage{soul}
\usepackage{threeparttable}
\usepackage[algo2e,ruled,vlined,linesnumbered]{algorithm2e}

\usepackage{booktabs}

\usepackage[colorinlistoftodos]{todonotes}

\ifpdf
  \DeclareGraphicsExtensions{.eps,.pdf,.png,.jpg}
\else
  \DeclareGraphicsExtensions{.eps}
\fi

\usepackage{enumitem}
\setlist[enumerate]{leftmargin=.5in}
\setlist[itemize]{leftmargin=.5in}

\newsiamremark{remark}{Remark}
\newsiamremark{hypothesis}{Hypothesis}
\crefname{hypothesis}{Hypothesis}{Hypotheses}
\newsiamthm{claim}{Claim}

\headers{Implied Volatility with VAEs}{Ning, Jaimungal, Zhang, Bergeron}

\title{Arbitrage-Free Implied Volatility Surface Generation\\ with Variational Autoencoders\thanks{The authors thank Ivan Sergienko for his comments on earlier versions of this work. S.J. acknowledges the support of the Natural Sciences \& Engineering Research council of Canada [ALLRP 550308 - 20].}
}

\author{Brian (Xin) Ning\thanks{Department of Statistical Sciences, University of Toronto (\email{brian.ning@mail.utoronto.ca},
\email{sebastian.jaimungal@utoronto.ca}; \url{http://sebastian.statistics.utoronto.ca},
\email{xiaorong.zhang@mail.utoronto.ca})}\and Sebastian Jaimungal\footnotemark[2]
\and Xiaorong Zhang\footnotemark[2]\and Maxime Bergeron\thanks{Riskfuel Analytics (\email{mb@riskfuel.com}; \url{http://riskfuel.com})}
}

\usepackage{amsopn}

\renewcommand{\arraystretch}{1.7}

\DeclareMathOperator*{\argmin}{argmin}
\DeclareMathOperator*{\argmax}{argmax}
\newcommand{\distas}[1]{\mathbin{\overset{#1}{\kern\z@\sim}}}
\newcommand{\norm}[1]{\big\lVert#1\big\rVert}

\newcommand{\K}{{\mathfrak{K}}}
\newcommand{\G}{{\mathfrak{G}}}
\newcommand{\Z}{{\mathfrak{Z}}}
\newcommand{\D}{{\mathfrak{D}}}
\newcommand{\M}{{\mathfrak{M}}}

\newcommand{\R}{{\mathds{R}}}
\newcommand{\Q}{{\mathbb{Q}}}
\newcommand{\E}{{\mathbb{E}}}
\newcommand{\x}{{\mathbf{x}}}
\newcommand{\z}{{\mathbf{z}}}
\newcommand{\y}{{\mathbf{y}}}
\newcommand{\EQ}{{\E^\Q}}
\newcommand{\mubar}{\bar{\mu}}
\newcommand{\tmu}{\tilde{\mu}}
\newcommand{\tpi}{\tilde{\pi}}
\newcommand{\tsig}{\tilde{\sigma}}
\newcommand{\tp}{\tilde{p}}

\newcommand{\Id}{{\mathds{1}}}
\newcommand{\new}[1]{#1}

\tikzstyle{reward}=[shape=circle,draw=blue!50,fill=blue!10]
\tikzstyle{action}=[shape=circle,draw=green,fill=green!10]
\tikzstyle{state}=[shape=circle,draw=red!50,fill=red!10]
\tikzstyle{gru}=[shape=rectangle,draw=black!50,fill=lime!10]
\tikzstyle{obs}=[shape=circle,draw=blue!50,fill=blue!10]
\tikzstyle{lightedge}=[<-,dotted]
\tikzstyle{mainstate}=[state,thick]
\tikzstyle{mainedge}=[<-,thick]

\tikzstyle{input} = [rectangle, minimum width=1cm, minimum height=4cm,text centered, draw=black, fill=green!30]
\tikzstyle{output} = [rectangle, minimum width=1cm, minimum height=4cm,text centered, draw=black, fill=blue!30]
\tikzstyle{latent} = [rectangle, minimum width=1cm, minimum height=1cm,text centered, draw=black, fill=red!30]
\tikzstyle{encoder} = [trapezium, trapezium left angle=70, trapezium right angle=70, minimum height=1cm, rotate=270, text centered, draw=black, fill=green!30]
\tikzstyle{decoder} = [trapezium, trapezium left angle=70, trapezium right angle=70, minimum height=1cm, rotate=90, text centered, draw=black, fill=blue!30]
\tikzstyle{arrow} = [thick,->,>=stealth]

\graphicspath{{./Figures/}}

\begin{document}

% \title{Arbitrage-Free Implied Volatility Surface Generation with Variational Autoencoders\thanks{%We would like to thank for the usage of data provided by Riskfuel Analytics.
% S.J. acknowledges the support of the Natural Sciences \& Engineering Research council of Canada through NSERC Alliance [ALLRP 550308 - 20].
% The authors would also like to thank Ivan Sergienko for his comments on ealier versions of this work.}}
% \author[1]{Brian (Xin) Ning}
% \author[1]{Sebastian Jaimungal}
% \author[1]{Xiaorong Zhang}
% \author[2]{Maxime Bergeron}
% \affil[1]{Department of Statistical Sciences, University of Toronto}
% %\affil[2]{Department of Statistical Sciences, University of Toronto}
% \affil[2]{Riskfuel Analytics,
% Toronto, ON, Canada}

\maketitle

\begin{abstract}
% The implied volatility surface is one of the key features used in the pricing and and hedging of derivatives. Due to the nature of real-world data, accurate interpolation and generation of said surface is needed. However balancing between the flexibility of a model-free generative model the interpretability of a financial model is often difficult.
We propose a hybrid method for generating arbitrage-free implied volatility (IV) surfaces consistent with historical data by combining model-free Variational Autoencoders (VAEs) with continuous time stochastic differential equation (SDE) driven models. We focus on two classes of SDE models: regime switching models and L\'evy additive processes. By projecting historical surfaces onto the space of SDE model parameters, we obtain a distribution on the parameter subspace faithful to the data on which we then train a VAE. Arbitrage-free IV surfaces are then generated by sampling from the posterior distribution on the latent space, decoding to obtain SDE model parameters,  and finally mapping those parameters to IV surfaces. \new{We further refine the VAE model by including conditional features  and demonstrate its  superior generative out-of-sample performance.}
\end{abstract}

\section{Introduction}

Modelling implied volatility (IV) surfaces %, or more specifically a collection of traded option prices,
in a manner that reflects historical dynamics while remaining arbitrage-free is a challenging open problem in finance. There are numerous approaches driven by stochastic differential equations (SDEs) that aim to do just so, including local volatility models  \cite{dupire1994pricing}, stochastic volatility models \cite{heston1993closed,hagan2002managing}, stochastic local volatility models \cite{said1999pricing}, jump-diffusion models \cite{cont2002calibration}, and regime switching models \cite{buffington2002regime}, among many others. Such approaches make specific assumptions on the dynamics of the underlying asset  and a choice of an equivalent martingale measure in order to avoid arbitrage. While these assumptions are not necessarily dynamically consistent with historical data, they do allow, e.g., pricing exotic derivatives via Monte Carlo or PDE methods.

An alternative to the SDE  approach is to use  non-parametric models to approximate IV surfaces directly without making assumptions on the underlying dynamics. For example, ML models such as support vector machines (SVMs) have been used to model such surfaces \cite{zeng2019online}. The issue of ensuring arbitrage-free surfaces is often tackled jointly during model fitting \cite{bergeron2021autoencoders} either through penalisation of arbitrage constraints \cite{ackerer2020deep} or by directly encoding them into the network architecture \cite{zheng2019gated}. These approaches, however, typically do not provide any guarantees and may not be arbitrage-free across the entire surface. A recent intriguing approach \cite{cohen2021arbitrage} is to reduce surfaces to arbitrage-free `factors' -- learned, e.g., through principal component analysis (PCA)  -- which can then be modeled using neural SDEs \cite{li2020scalable}. This approach, while very promising, relies on the quality of the `factors' which are often complicated to compute. Another recent approach is that of \cite{doi:10.1137/20M1381538} where the authors use Gaussian processes under shape constraints to generate surfaces and illustrate good fits to S\&P data. Here, however, we are interested in the  setting of sparse FX data and in generating the distribution over surfaces in a manner that is consistent with the historical data. \new{The construction of arbitrage-free models based on ML  approaches for stochastic interest rates has been tackled in \cite{kratsios2020deep}. In contrast, our focus is on European options and, more specifically, our application setting is to FX options.}

%\seb{stick to notation SDE model and VAE model... and the overarching approach.}
In this paper, we develop a hybrid approach to resolve these issues by using SDE models that are by construction arbitrage-free yet flexible enough to fit arbitrary IV surfaces.  One immediate dividend of this approach lies in its ability to produce realistic synthetic training data that can be used to leverage deep learning pricing methods in downstream tasks \cite{ferguson2018deeply,horvath2019deep}. The class of SDE models we consider include time-varying regime switching models and L\'evy additive processes detailed in Section \ref{sec:models}. We avoid overfitting by incorporating a Wasserstein penalty to keep the SDE model's risk-neutral density from deviating too far from the \new{candidate} one. The SDE model parameters, once fitted to data, represent a parameter subspace reflecting the features embedded in the data. The distribution on the subspace depends on the characteristics of the underlying asset and can be complex. We ``learn'' this distribution by using Variational Autoencoders (VAEs) which also allows for disentanglement of the subspace in an interpretable manner. SDE model parameters may be generated from the VAE model and used to create IV surfaces that are both faithful to the historical data but also strictly risk-neutral. This is similar in spirit, but distinct from, the tangent L\'evy model approach introduced in \cite{carmona2012tangent} where a L\'evy density is used to generate arbitrage-free prices while, here, the VAE generates parameters of the SDE model.
%\maxime{This paragraph is potentially confusing, there is no clear distinction between the SDE piece of the model and the VAE used to work with the model's parameters. Perhaps the key is to distinguish the ``approach'' or ``algorithm'' which blends the two stepd from the word ``model''?}

%\maxime{KDE is used in Figure 1 before the acronym is introduced later in the text}

The overall approach may be summarised as: (i) fit a rich arbitrage-free SDE model to historical market data to obtain a collection of parameters, (ii) train a generative model, in particular a VAE model, on the collection of SDE model parameters, (iii) sample from the latent space of the generative VAE model, (iv) decode the samples to obtain a collection of SDE model parameters, and (v) use said SDE model and parameters to obtain arbitrage-free surfaces faithful to the historical data.  A flow-chart of the process is presented in Figure \ref{fig:flow_chart}.
\new{We further refine the VAE model by including conditioning features into the encoding and decoding architectures. This results in a conditional VAE (CVAE) model, first introduced in \cite{sohn2015learning} in a very different setting, for the arbitrage-free model parameter embeddings. We find that the CVAE model outperforms all others when comparing out-of-sample performance.}
%\maxime{It looks like some of the boxes in the Figure should be bigger to accommodate the text.}
% In this paper we present a new approach that attempts to combine both approaches in a manner where the flexibility of non-parametric models are preserved but a strict adherence to risk-neutrality is also maintained. By first fitting an over-parameterized model to the data while ensuring a smooth asset price model density, we are able to maintain the flexibility. Afterwards, we apply a non-parameterized machine learning model, in this case a variational auto-encoder, to the fitted parameter space in order to generate and predict new surfaces that are arbitrage free as the original model is by definition arbitrage free.

The remainder of this article is organised as follows. Section \ref{sec:Fitting} describes a generic method of fitting SDE models to a limited data set. Section \ref{sec:models} defines the financial models we use in calibration. Section \ref{sec:vae} details the structure of the VAE and its generative process. \new{Section \ref{sec:CVAE} extends the VAE framework by conditioning prespecified features.} Finally, Section \ref{sec:results} presents the results of our algorithm applied to 1,900 days of foreign exchange (FX) data for three currency pairs\footnote{AUD = Australian Dollar, USD = US Dollar, and CAD = Canadian Dollar.}: AUD-USD, EUR-USD, and CAD-USD.

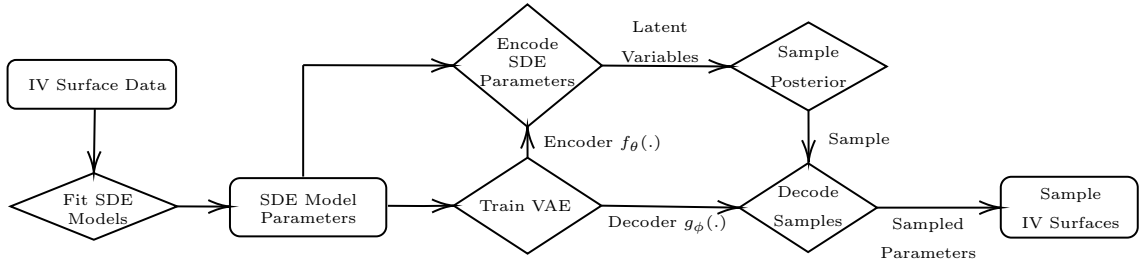
\begin{figure}
    \centering

\tikzset{every picture/.style={line width=0.75pt}} %set default line width to 0.75pt

\begin{tikzpicture}[x=0.75pt,y=0.75pt,yscale=-1,xscale=1]
%uncomment if require: \path (0,147); %set diagram left start at 0, and has height of 147

%Straight Lines [id:da6408402969191966]
\draw    (45,62.25) -- (44.54,92.35) ;
\draw [shift={(44.51,94.35)}, rotate = 270.88] [color={rgb, 255:red, 0; green, 0; blue, 0 }  ][line width=0.75]    (10.93,-3.29) .. controls (6.95,-1.4) and (3.31,-0.3) .. (0,0) .. controls (3.31,0.3) and (6.95,1.4) .. (10.93,3.29)   ;
%Flowchart: Decision [id:dp3226868473822677]
\draw   (44.51,94.35) -- (86.52,111.12) -- (44.51,127.9) -- (2.5,111.12) -- cycle ;
%Flowchart: Alternative Process [id:dp14623664672596104]
\draw   (1,41.74) .. controls (1,39.37) and (2.92,37.45) .. (5.3,37.45) -- (81.7,37.45) .. controls (84.08,37.45) and (86,39.37) .. (86,41.74) -- (86,57.7) .. controls (86,60.08) and (84.08,62) .. (81.7,62) -- (5.3,62) .. controls (2.92,62) and (1,60.08) .. (1,57.7) -- cycle ;
%Straight Lines [id:da5037538950256744]
\draw    (86.52,111.12) -- (113.16,111.23) ;
\draw [shift={(113.16,111.23)}, rotate = 180] [color={rgb, 255:red, 0; green, 0; blue, 0 }  ][line width=0.75]    (10.93,-3.29) .. controls (6.95,-1.4) and (3.31,-0.3) .. (0,0) .. controls (3.31,0.3) and (6.95,1.4) .. (10.93,3.29)   ;
%Flowchart: Alternative Process [id:dp37132294366795926]
\draw   (113.16,102.02) .. controls (113.16,99.19) and (115.45,96.9) .. (118.29,96.9) -- (186.87,96.9) .. controls (189.7,96.9) and (192,99.19) .. (192,102.02) -- (192,121.07) .. controls (192,123.9) and (189.7,126.2) .. (186.87,126.2) -- (118.29,126.2) .. controls (115.45,126.2) and (113.16,123.9) .. (113.16,121.07) -- cycle ;
%Flowchart: Decision [id:dp044902232477203485]
\draw   (263.25,86.61) -- (300.5,110.75) -- (263.25,134.89) -- (226,110.75) -- cycle ;
%Straight Lines [id:da6707279255450451]
\draw    (192.83,110.83) -- (210.9,110.94) -- (224,110.78) ;
\draw [shift={(226,110.75)}, rotate = 539.27] [color={rgb, 255:red, 0; green, 0; blue, 0 }  ][line width=0.75]    (10.93,-3.29) .. controls (6.95,-1.4) and (3.31,-0.3) .. (0,0) .. controls (3.31,0.3) and (6.95,1.4) .. (10.93,3.29)   ;
%Flowchart: Decision [id:dp5124280495825699]
\draw   (263.2,9.51) -- (300.77,40.28) -- (263.2,71.06) -- (225.64,40.28) -- cycle ;
%Straight Lines [id:da10957724686193271]
\draw    (150.08,40.07) -- (223.64,40.28) ;
\draw [shift={(225.64,40.28)}, rotate = 180.16] [color={rgb, 255:red, 0; green, 0; blue, 0 }  ][line width=0.75]    (10.93,-3.29) .. controls (6.95,-1.4) and (3.31,-0.3) .. (0,0) .. controls (3.31,0.3) and (6.95,1.4) .. (10.93,3.29)   ;
%Straight Lines [id:da3976240728434737]
\draw    (150.08,40.07) -- (150.08,96.12) ;
%Straight Lines [id:da27736170669820925]
\draw    (263.25,86.61) -- (263.21,73.06) ;
\draw [shift={(263.2,71.06)}, rotate = 449.83] [color={rgb, 255:red, 0; green, 0; blue, 0 }  ][line width=0.75]    (10.93,-3.29) .. controls (6.95,-1.4) and (3.31,-0.3) .. (0,0) .. controls (3.31,0.3) and (6.95,1.4) .. (10.93,3.29)   ;
%Straight Lines [id:da745513330915444]
\draw    (300.77,40.28) -- (363.5,40.74) ;
\draw [shift={(365.5,40.75)}, rotate = 180.41] [color={rgb, 255:red, 0; green, 0; blue, 0 }  ][line width=0.75]    (10.93,-3.29) .. controls (6.95,-1.4) and (3.31,-0.3) .. (0,0) .. controls (3.31,0.3) and (6.95,1.4) .. (10.93,3.29)   ;
%Flowchart: Decision [id:dp39290714517766956]
\draw   (404.94,18.58) -- (444.38,40.75) -- (404.94,62.92) -- (365.5,40.75) -- cycle ;
%Straight Lines [id:da1231622560011505]
\draw    (404.94,62.92) -- (405,86.75) ;
\draw [shift={(405,88.75)}, rotate = 269.87] [color={rgb, 255:red, 0; green, 0; blue, 0 }  ][line width=0.75]    (10.93,-3.29) .. controls (6.95,-1.4) and (3.31,-0.3) .. (0,0) .. controls (3.31,0.3) and (6.95,1.4) .. (10.93,3.29)   ;
%Flowchart: Decision [id:dp34190687709590306]
\draw   (404.75,89.39) -- (439.5,111.75) -- (404.75,134.11) -- (370,111.75) -- cycle ;
%Straight Lines [id:da09204641083934706]
\draw    (300.5,110.75) -- (368,111.72) ;
\draw [shift={(370,111.75)}, rotate = 180.82] [color={rgb, 255:red, 0; green, 0; blue, 0 }  ][line width=0.75]    (10.93,-3.29) .. controls (6.95,-1.4) and (3.31,-0.3) .. (0,0) .. controls (3.31,0.3) and (6.95,1.4) .. (10.93,3.29)   ;
%Flowchart: Alternative Process [id:dp14205985540193988]
\draw   (501.92,101.59) .. controls (501.92,98.63) and (504.32,96.23) .. (507.28,96.23) -- (565.64,96.23) .. controls (568.6,96.23) and (571,98.63) .. (571,101.59) -- (571,121.48) .. controls (571,124.44) and (568.6,126.83) .. (565.64,126.83) -- (507.28,126.83) .. controls (504.32,126.83) and (501.92,124.44) .. (501.92,121.48) -- cycle ;
%Straight Lines [id:da5966139170697711]
\draw    (439.5,111.75) -- (500,111.75) ;
\draw [shift={(502,111.75)}, rotate = 180] [color={rgb, 255:red, 0; green, 0; blue, 0 }  ][line width=0.75]    (10.93,-3.29) .. controls (6.95,-1.4) and (3.31,-0.3) .. (0,0) .. controls (3.31,0.3) and (6.95,1.4) .. (10.93,3.29)   ;

% Text Node
\draw (46.35,49.95) node  [font=\tiny] [align=left] {{\tiny IV Surface Data}};
% Text Node
\draw (24.16,107.21) node [anchor=north west][inner sep=-2pt]
[font=\tiny] [align=left]
%{{\tiny Fit SDE Model}};
{\begin{minipage}[lt]{36.91pt}\setlength\topsep{0pt}
\begin{center}
{\tiny Fit SDE }\\{\tiny Models}
\end{center}
\end{minipage}};
% Text Node
\draw (262.03,110.75) node  [font=\footnotesize] [align=left] {{\tiny Train VAE}};
% Text Node
\draw (263.2,38.28) node  [font=\tiny] [align=left] {\begin{minipage}[lt]{36.91pt}\setlength\topsep{0pt}
\begin{center}
{\tiny Encode SDE }\\{\tiny Parameters}
\end{center}

\end{minipage}};
% Text Node
\draw (269.85,73.96) node [anchor=north west][inner sep=0.75pt]  [font=\footnotesize] [align=left] {{\tiny Encoder $f_\theta(.)$}};

% Text Node
\draw (302.27,16.63) node [anchor=north west][inner sep=0.75pt]  [font=\footnotesize] [align=left] {\begin{minipage}[lt]{39.43pt}\setlength\topsep{0pt}
\begin{center}
{\tiny Latent }\\{\tiny Variables}
\end{center}

\end{minipage}};
% Text Node
\draw (404.94,40.75) node  [font=\footnotesize] [align=left] {\begin{minipage}[lt]{38.59pt}\setlength\topsep{0pt}
\begin{center}

{\tiny Sample}\\{\tiny Posterior}
\end{center}

\end{minipage}};

% Text Node
\draw (404.75,111.75) node  [font=\footnotesize] [align=left] {\begin{minipage}[lt]{38.59pt}\setlength\topsep{0pt}
\begin{center}
{\tiny Decode}\\{\tiny Samples}
\end{center}

\end{minipage}};
% Text Node
\draw (411.26,72.46) node [anchor=north west][inner sep=0.75pt]  [font=\footnotesize] [align=left] {\begin{minipage}[lt]{26.55pt}\setlength\topsep{0pt}
\begin{center}
{\tiny Sample}
\end{center}

\end{minipage}};
% Text Node
\draw (302.32,114.06) node [anchor=north west][inner sep=0.75pt]  [font=\footnotesize] [align=left]{{\tiny Decoder $g_\phi(.)$}};

% Text Node
\draw (508.46,100.25) node [anchor=north west][inner sep=0.75pt]  [font=\footnotesize] [align=left] {\begin{minipage}[lt]{40.81pt}\setlength\topsep{0pt}
\begin{center}
{\tiny Sample}\\{\tiny IV Surfaces}
\end{center}

\end{minipage}};
% Text Node
\draw (439.79,115.47) node [anchor=north west][inner sep=0.75pt]  [font=\footnotesize] [align=left] {\begin{minipage}[lt]{36.06pt}\setlength\topsep{0pt}
\begin{center}
{\tiny Sampled}\\{\tiny Parameters}
\end{center}

\end{minipage}};
% Text Node
% Text Node
\draw (152.7,112.28) node  [font=\tiny] [align=left] {\begin{minipage}[lt]{36.81pt}\setlength\topsep{0pt}
\begin{center}
{\tiny SDE Model}\\{\tiny Parameters}
\end{center}

\end{minipage}};

\end{tikzpicture}
    \caption{Flow chart of algorithm starting from raw IV surface data to generated surfaces.\label{fig:flow_chart}}
\vspace*{-2em}
\end{figure}
%\maxime{In Figure 1 we should the distinction between VAE model and SDE model.}
\section{Model Setup and Estimation Procedure}
\label{sec:Fitting}

We work with a completed filtered probability space  $(\Omega,\mathbb{Q},\mathcal{F},(\mathcal{F}_t)_{t\geq0})$ where the filtration is the natural one generated by a stochastic driver $X:=(X_t)_{t\geq0}$. We explore several choices of models for $X$ in Section \ref{sec:models}. Here, $\mathbb{Q}$ represents the risk-neutral probability measure and we assume that the market prices options using this measure and model the FX rate process $S=(S_t)_{t\ge0}$  as follows:
%\seb{Is the discounting fully correct here?}
%\brian{Yes, rather than piece-wise across time we defined it so that the discounting across the entire duration of time from t = 0 to $\tau$ is the rate ${r^d} - {r^f}$. In practice we assume $r_d = r_f = 0$ but the modelling should work in general.}
\begin{equation}
    S_t = S_0\, e^{\int_0^t (r^d_s - r^f_s)\,ds + \int_0^{t} l(s)ds + X_t }, \qquad \forall \; t\ge0,
\end{equation}
where $r^d:=(r^d_t)_{t\ge0}$ and $r^f:=(r^f_t)_{t\ge0}$ are the domestic and foreign short rate processes, and $l$ is a deterministic function of time that ensures $(e^{-\int_0^t (r^d_s - r^f_s)\,ds } S_t)_{t\ge0}$ is a $\mathbb{Q}$-martingale. We assume interest rates are  deterministic since there is no conceptual difficulty in generalising to the stochastic case.

From, e.g., Theorem 3.2 in \cite{lewis2001simple}, we may write the undiscounted option price as
\begin{equation}
\EQ \Big[ (S_\tau - K )_+\Big]
    =
    \tilde{S}_0 - \frac{\sqrt{K \,\tilde{S}_0}}{\pi} \int_{0}^{\infty} \mathfrak{R}\!\left( e^{i\,z\,\log \frac{\tilde{S}_0}{K}}
    \,\phi_\theta\left(z - \tfrac{1}{2}i\right)\right) \frac{dz}{z^2 + \frac{1}{4}},
    \label{eqn:FourierTransformPrice}
\end{equation}
where $\tilde{S}_0:=S_0\, e^{\int_0^\tau (r^d_s - r^f_s)\,ds + \int_0^{\tau} l(s)ds}$, $\phi_\theta(z):=\EQ[e^{izX_\tau}]$ is the  characteristic function of $X_\tau$, $\theta$ encodes the parameters of the stochastic driver $X$, and $\mathfrak{R}(\cdot)$ denotes the real component of its argument. For the class of SDE models considered here, the characteristic function is known in closed form and the above formula allows for efficient calibration to market data.
A naive approach to parameter estimation is to minimise the squared error between the model and data prices.
%$\argmin_\theta \norm{\boldsymbol{C^m_\theta} - \boldsymbol{C^d}}^2\!\!$,
Such parameter estimation is prone to overfitting when data is sparse as is often the case in FX markets. To address this issue, we add a regularisation term to our objective. Specifically, we use the $1$-Wasserstein distance between the model price's risk-neutral probability distribution function (pdf), denoted $F^m_\theta(dx)$, and the implied pdf derived from option data, denoted $F^d(dx)$. Wasserstein distances provide a natural metric on the space of probability measures and have seen wide application across many fields \cite{villani2009optimal}. Other choices include divergences, such as the Kullback-Liebler divergence, or metric variations such as Jensen-Shannon divergence. We chose the Wasserstein distance in particular because it is the most ubiquitous and robust.
%Here, it is natural to use this distance to regularise model parameters, akin to how Lasso \cite{tibshirani1996regression} regularises regression models.
%\maxime{I can sense a reader asking themselves: why is Wasserstein distance a good choice here?}
% Need to find reference for second deriv

To this end, the 1-Wasserstein distance between $F^d$ and $F^m_\theta$ is given by
\begin{equation}
    W_1(F^d, F^m_\theta) = \inf_{\pi \in \Pi(F^d, F^m_\theta)}\textstyle\int_{\mathds{R}^2} |x - y|\, d\pi(x,y)
    = \textstyle\int_\R |F^d(x)-F^m_\theta(x)|\,dx,
\label{eqn:wass}   %
\end{equation}%
where $\Pi(F^d, F^m_\theta)$ denotes the set of all probability distributions on $\mathds{R}^2$ with marginals $F^d$ and $F^m_\theta$ and the second equality holds in dimension one \cite{vallender1974calculation}.
As data is observed at discrete strikes, we approximate the \new{candidate} density by interpolating IVs at each fixed maturity using B-splines. It is well known \cite{breeden1978prices} that $\partial_{KK} C(T,K)$ corresponds to the risk-neutral density of the underlying asset price evaluated at $K$. \new{The derivation of the spline implied density for the case of call options can be found in Appendix \ref{sec:ERDens}. }
%As a corollary, you can easily show that taking the second derivative of the normalised call price ratio $Pr^d$ indeed gives the density of the asset price evaluated at the normalised strike $K$.
Since we are using a B-spline, the corresponding density is not necessarily risk-neutral. This is not a problem, however, as we merely use them as a regularisation term and ultimately use risk-neutral models to derive option prices.

%\subsection{Loss Function}
We use a combination of the pricing error and the Wasserstein distance \eqref{eqn:wass} as the loss function in model estimation. That is we seek to obtain model parameters
\begin{equation}
\label{eqn:fitting_loss}
    \theta^*:=\argmin_\theta \left( \norm{\boldsymbol{C^m_\theta} - \boldsymbol{C^d}} ^ 2 + \alpha \textstyle\sum_{n=1}^N W_1(F^{d,(n)}, F^{m,(n)}_\theta) \right)\!,
\end{equation}
where $\boldsymbol{C^d}$ and $\boldsymbol{C^m_\theta}$ denotes the flattened vector of data and model prices, respectively, at each strike-maturity pair, and $F^{d,(n)}$ and $F^{m,(n)}_\theta$ represent the model and data implied distribution functions at each maturity, respectively.
The regularisation parameter $\alpha \geq 0$  controls the importance placed on being close to the spline implied densities. \new{The effect this hyperparameter has on model accuracy is explored in detail in Appendix \ref{sec:alpha_eff}.}

\section{Class of Stochastic Drivers}
\label{sec:models}

In this section, we describe the class of models over which we perform estimation. Throughout, we denote the sequence of dates on which we have option implied volatility data by $\{\tau_1,\tau_2,\dots,\tau\}$ and define $\tau_0=0$.

\subsection{CTMC}
\label{sec:CTMC}

% \seb{how about just calling these regime switching models? RS.. as CTMC only describes the regimes and not the full model?}

The continuous time Markov-Chain model (CTMC) is a multi-regime model that assumes the underlying asset follows a Geometric Brownian motion (GBM) modulated by a continuous time Markov-Chain representing the current market regime. They were first introduced into financial modelling in \cite{buffington2002regime}. Here, however, we generalise the model to allow for time-varying parameters and use transform methods in \cite{jackson2008fourier} to solve for the characteristic function. We use this model as a non-parametric approach to modelling the sequence of risk-neutral densities.  The potential of overfitting of such models is mitigated by the Wasserstein distance penalty in \eqref{eqn:wass}.

%\seb{give some rationale as to why we use this... non-parametric version of fitting a density}

Let $(Z_t)_{t\ge0}$ denote a continuous time Markov chain taking on values in $\K:=\{1 \dots K\}$.  Suppose, moreover, that the generator matrix $A$ driving the CTMC, the regime specific vector of drifts $\mu$, and the regime specific vector of volatilities $\sigma$ are all constant on the sequence of maturity intervals $[\tau_{i-1},\tau_{i})$, but may vary across maturity periods. We can then consider a driving process $X$ satisfying the SDE
\begin{equation}
    dX_t = \left( \mu^{(n)}_{Z_t} - \tfrac{1}{2} (\sigma^{(n)}_{Z_t})^2 \right) \, dt
    + \sigma^{(n)}_{Z_t}\, dW_t, \qquad \forall t\in[\tau_{n-1},\tau),
    \label{eqn:dX-CTMC}
\end{equation}
where $\mu_{k}^{(n)}\in\R$ and $\sigma_k^{(n)}\in\R_+$ denote the expected return and instantaneous volatility in period $n$ when $Z_t=k$, $k\in\K$. Similarly, we let $A^{(n)}$  denote the transition rate matrix for period $n$, satisfying $A^{(n)}_{ij}\in\R_+$, $i\ne j$, and $\sum_{j=1}^KA_{ij}=0$.
\begin{proposition}
\label{prop:ctmc}
   \!\!If $X$ satisfies \eqref{eqn:dX-CTMC}, then the characteristic function $\phi_{X_{\tau}}\!(\omega)\!:=\!\EQ[e^{i\omega\,X_{\tau}}]$ is given by
    $\phi_{X}(\omega) = {\boldsymbol{\pi}}^\intercal e^{\tau_1 \Psi^{(1)}(\omega)} e^{(\tau_{2} - \tau_{1}) \Psi^{(2)}(\omega)} \dots e^{(\tau - \tau_{N-1}) \Psi^{(N)}(\omega)}\; \boldsymbol{1},
    $
    where $\pi_k=\mathbb{Q}(Z_0=k)$ is the prior probability of the latent state,
    \begin{equation}
    [\Psi^{(n)}(\omega)]_{jk} := \left( A^{(n)}_{kk} + i \left(\mu_k^n - \tfrac{1}{2} (\sigma_k^n)^2\right) \omega - \tfrac{1}{2} (\sigma_k^n)^2 \omega^2 \right) \delta_{jk} + A^{(n)}_{jk} \;(1 - \delta_{jk}),
\end{equation}
and $\delta_{jk}$ is the Kroencker delta, which equals $1$ if $j = k$ and $0$ otherwise.
\end{proposition}
\begin{proof}
    See Appendix
\end{proof}

%\subsection{Structure of the A matrix}

% TO MOVE (likely in results/implementation section along with brief description of solving for parameters iteratively)

%We assume that each state k has its corresponding rate parameter $\lambda_k$ which determines the rate at which the chain will move out of the state. For any state $1 \leq k \leq K $, the chain has an equal chance of moving into the state below ($k-1$) or above ($k + 1$). We assume that state -1 is equivalent to state K and thus can one can transition directly between state one and state K. Any other transitions will have probability 0. This restricts the process to only being able to move through one state at a time without jumping.

% This creates the generator matrix of the form:
% \begin{equation}
%     A^n =
%     \begin{bmatrix}
%     -\lambda_1^n & \frac{\lambda_1^n}{2} & 0 & \dots &\frac{\lambda_1^n}{2}\\
%     \frac{\lambda_2^n}{2} & -\lambda_2^n & \frac{\lambda_2^n}{2} & 0 & \dots \\
%     \vdots & \ddots\\
%     \dots & 0 & \frac{\lambda_{K-1}^n}{2} & -\lambda_{K-1}^n & \frac{\lambda_{K-1}^n}{2} \\
%     \frac{\lambda_K^n}{2} & \dots & 0 & \frac{\lambda_K^n}{2} & -\lambda_K^n
%     \end{bmatrix}
% \end{equation}

% which limits the number of parameters needed to be estimated at each given maturity to $3K$ (K from $\sigma$, K from $\mu$, and K from $\lambda$) excluding the vector of initial probabilities, where $K$ is the number of possible states of the system.

To assist with identifiability when estimating the CTMC model parameters, we introduce a cyclic  structure on the transition rate matrices. More precisely, we require that  $A^{(n)}_{ij}=\frac12\lambda_i^{(n)}(\delta_{j=i+1}+\delta_{j=i-1})$, for all $K>i>1$,    $A^{(n)}_{1,2}=A^{(n)}_{1,K}=\frac12\lambda^{(n)}_1$, $A^{(n)}_{K,1}=A^{(n)}_{K,K-1}=\frac12\lambda^{(n)}_K$, and $\lambda^{(n)}_i>0$, together with the usual constraint that $\sum_{j=1}^KA_{ij}=0$, for all $i\in\K$.

Figure \ref{fig:ex_plots} shows the CTMC model fitted to two days of data including the corresponding implied densities. \new{The two specific days are chosen as they correspond to a root mean squared error (rmse) that lie in the 50th and 90th quantiles of rmse across all days. }

\begin{figure}
\centering
\includegraphics[width=0.8\textwidth]{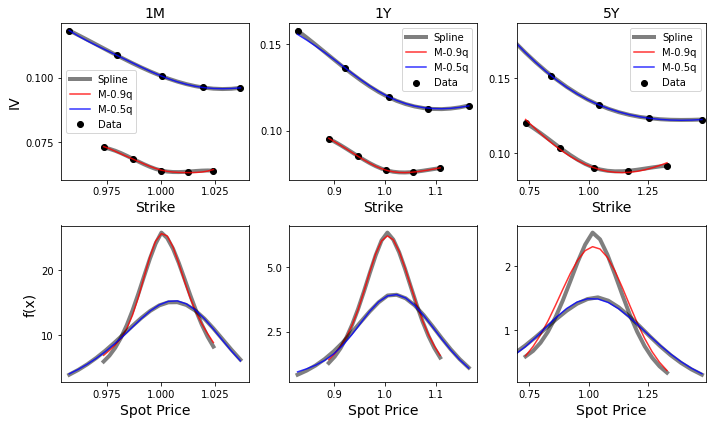}
\vspace*{-1em}
\caption{\new{Typical fits of the CTMC model to the IVs (top) and densities (bottom) of AUD-USD data with root mean squared error in the 50th (blue) and 90th (red) quantiles. Strike ranges differ on different days.}}
\label{fig:ex_plots}
\vspace*{-1em}
\end{figure}

\subsection{L\'evy Additive Processes}
\label{sec:jm}

For comparison, we also study a class of L\'evy additive processes to allow jumps in FX rates. To this end, we model $X$ as
\begin{equation}
    X_t = \textstyle\int_0^t \textstyle\int_\R y\, \left[\mu(dy,ds)-\nu(dy,ds)\right]+\textstyle\int_0^t \sigma_s\,dW_s,
    \label{eqn:X-additive-process}
\end{equation}
where $(\sigma_t)_{t\ge0}$ is piecewise deterministic, $\sigma_t :=\sigma^{(n)} \Id_{t\in[\tau_{n-1},\tau)}$, $\mu$ is a Poisson random measure with compensator $\nu$, and where $\nu(dy,ds)=\Id_{t\in[\tau_{n-1},\tau)}\,\nu^{(n)}(dy)\,ds$ with $\nu^{(n)}(dy)$ being L\'evy measures. This allows the structure of the L\'evy measure to differ between maturity periods and allows for both finite and infinite activity processes. For example, we may have  $\nu^{(n)}(dy)=\lambda^{(n)}\,F^{(n)}(dy)$, in which case $X$ is an additive compound Poisson process with jump measure $F^{(n)}(dy)$ and intensity $\lambda^{(n)}$, or   $\nu^{(n)}(dy)=C\left(\frac{e^{-G|y|}}{|y|^{1+Y}}\Id_{x<0} + \frac{e^{-M\,y}}{|y|^{1+Y}}\Id_{x>0}\right)\,dy$ in which case $X$ is additive version of a tempered stable L\'evy measure (also known as the KoBol \cite{boyarchenko2000option,boyarchenko2002non} or CGMY \cite{carr2002fine} models). There are a slew of alternate models as well, however, in the sake of brevity we restrict to these two classes. For the additive compound Poisson process, we include two jump measure types: mixture of normals $F^{(n)}(dy) = \sum_{j=1}^K\tilde{\pi}_k^{(n)}\,\phi\left(\frac{y-\tilde{\mu}_k^{(n)}}{\tilde{\sigma}_k^{(n)}}\right)\,dy$, where $\phi$ is the standard normal pdf -- to mimick a non-parametric estimation of the jump distribution implied by the data, and a double exponential model (see \cite{kou2004option}), in which case $F^{(n)}(dy)=\left((1-p^{(n)})\,e^{-a_-^{(n)}|x|}\Id_{x<0} +p^{(n)}\,e^{-a_+^{(n)}x}\Id_{x>0} \right)dy$.
\begin{proposition}
\label{prop:LevyAdditive}
    If $X$ satisfies \eqref{eqn:X-additive-process}, then the characteristic function
$\phi_{X}(\omega):=\EQ[e^{i\omega\,X_{\tau}}]$ is given by $\phi_{X}(\omega)=e^{\sum_{n=1}^N \Psi^{(n)}(\omega)(\tau-\tau_{n-1})}$,
         where
        \begin{equation}
        \Psi^{(n)}(\omega)=-\tfrac12\, (\sigma^{(n)})^2\,\omega^2 + \int_\R \left(e^{i\omega y}-1-i\omega y \Id_{|y|\le 1}\right)\nu^{(n)}(dy)\,.
        \label{eqn:Levy-char}
        \end{equation}
%\end{subequations}
\end{proposition}
\begin{proof}
    Apply the L\'evy-Khintchine formula \cite[Chap 1.2.4]{applebaum2009levy} within each period.
\end{proof}
The specific form of the L\'evy characteristic function appearing in \eqref{eqn:Levy-char} for the models we employ in the numerical analysis appear in Table \ref{tab: char}. The parameters for the appropriate period should be inserted into these expression when computing the full characteristic function. We also record the characteristic function for the CTMC model in the same table.

% \begin{equation}
%     X_{t}  = \sum_{j=1}^{k} (\mu_j(\tau_j-\tau_{j-1}) +  \sigma_j( W_{\tau_j}-W_{\tau_{j-1}})) + \sum_{n=1}^{N_{\tau_k}}\xi
%     \label{eqn:jump_process}
% \end{equation}
% where $\mu_j$ and $\sigma_j$ denote the drift and diffusion of the Geometric Brownian motion in period $j$, and $\xi$ denotes the jump at Poisson time $N_{\tau_{k}} \stackrel{\Q}{\sim} Poisson(\lambda \tau_{k})$.

%Probably need reference on dbl-exp? Kou

% We implemented two jump-diffusion models: the Double Exponential JD and the Mixture of Gaussian JD. Their characteristic functions are provided in table \ref{tab: char}.

% \subsection{Pure Jump Model}
% %Reference: Jackson & Jaimungal
% Different from the jump-diffusion model, the pure jump models have jumps occurring infinitely often with most jumps being of infinitesimal size, which makes it capable of calibrating both small and large movement.
% %Reference: the original paper of CGMY.

% We implemented the CGMY extension of the Variance-Gamma model and its characteristic function is provided in Table \ref{tab: char}.

\begin{table}
\footnotesize
\renewcommand\arraystretch{1.1}
\centering
\begin{tabular}{rl}
 \toprule
 \toprule
%  \multicolumn{2}{c}{Characteristic Functions of Models}
%  \\
%  \midrule
 Model & Characteristic Function\\
 \midrule
 CTMC   & ${\boldsymbol{\pi}}^\intercal e^{\tau_1 \Psi^{(1)}(\omega)} e^{(\tau_{2} - \tau_{1}) \Psi^{(2)}(\omega)} \dots e^{(\tau - \tau_{N-1}) \Psi^{(N)}(\omega)} \boldsymbol{1}$
 \\
 Double Exponential JD& $-\frac{\sigma^2\omega^2}{2}+\lambda\left(p\frac{a_+}{a_+-i\omega} + (1-p)\frac{a_-}{a_-+i\omega}-1\right)$\\
 Gaussian Mixture JD& $-\frac{\sigma^2\omega^2}{2}+\lambda\left(\sum_{i=1}^{K}\tilde{\pi}_i(e^{i\tilde{\mu_i}\omega-\tilde{\sigma_i}^2\omega^2/2})-1\right)$ \\
 CGMY/KoBoL & $C\Gamma(-Y)[(M-i\omega)^Y-M^Y+(G+i\omega)^Y-G^Y]$ \\
 \toprule
 \toprule
\end{tabular}
\vspace*{-1em}
\caption{Summary of characteristic functions within a given maturity period.\label{tab: char}
~\vspace*{-2em}}
\end{table}

\section{Model Parameter Generation}

In the previous section, we described two classes of stochastic models that can be calibrated to option prices. The SDE model parameters $\boldsymbol{\theta}$ are estimated by minimising a weighted average of the mean-squared error in option prices and the Wasserstein distance between the model implied risk-neutral densities and the densities implied by a $B$-spline fit of the IV smiles.  Once the SDE model parameters are estimated on market data, our goal is to generate new synthetic parameters $\boldsymbol{\theta}$ that are consistent with the historical data. This allows us to produce synthetic IV surfaces that are guaranteed to be both arbitrage-free and representative of real surfaces. %The specific generative model we employ are variational autoencoders (VAEs). The next subsections provides a brief overview of VAEs, how they are trained, and how we sample from a trained model.
%\maxime{Using the word learned here feels odd. I would suggest replacing it with the word ``trained'' here and throughout the text. }

\subsection{Variational Autoencoder (VAE)}
\label{sec:vae}

Variational Autoencoders \cite{kingma2013auto}  are generative models that aim to train a multivariate latent representation, known as an encoding, from a collection of data. \new{A key advantage of VAEs over conventional dimensionality reduction techniques such as vanilla autoencoders (AEs), PCA, or KPCA is their ability to generalize the latent feature space. A well-known issue of AEs is that the latent space may not be continuous and may not exhibit any well defined structure. This makes it difficult to interpolate between training data points and poses difficulties in generating new data, as demonstrated by  \cite{oring2020autoencoder}. In \cite{oring2020autoencoder}, attempts to rectify such issues are made by regularizing the training procedure to ensure the latent manifold is smooth and locally convex. VAEs, however, place a prior on the latent feature and are thus able to easily generate out-of-sample data by sampling from either the prior or the posterior, depending on the specific application. } 

Given a set of samples $\{\mathbf{x_i}\,|\,\mathbf{x_i} \in \R^D \}_{i\geq1}$ from a distribution parameterized by ground truth latent factors $\{\mathbf{z_i}\,|\,\mathbf{z_i} \in \R^{D'} \}_{i\geq1}$ where $D \gg D'$ and a generative model $p_{\psi}(\x)$, we seek to maximise the log-likelihood $\log \, p_{\psi}(\x)$. The log-likelihood is, however, intractable as it involves integration over the posterior $p_{\psi}(\z|\x)$, which, even in setups where $p(\x|\z)$ is specified, is itself intractable. The VAE circumvents this issue by introducing an approximation of the true posterior with a neural network $q_{\phi}(\z|\x)$ parameterized by $\phi$. Indeed, for any distribution $q_\phi(\z|\x)$, we have that the log-likelihood satisfies the inequality
\begin{subequations}
{
\begin{align*}
    \log p_\psi(\x)
    &= \log \int p_\psi(\x|\z)\;p_\psi(\z)\,d\z
    = \log \int \frac{p_\psi(\x|\z)\;p_\psi(\z)}{q_\phi(\z|\x)}\,q_\phi(\z|\x)\,d\z
    \\
    &\ge \int \log\left(\frac{p_\psi(\x|\z)\;p_\psi(\z)}{q_\phi(\z|\x)}\right) \,q_\phi(\z|\x)\,d\z
    = \E_{q_\phi(\z|\x)}\left[ \log p_\psi(\x|\z) \right] - KL\left[q_\phi(\z|\x)\,\|\,p_\psi(\z)\right]\,.
\end{align*}%
}%
\end{subequations}%
The right most expression of this inequality is known as the evidence lower bound (ELBO) of the log-likelihood. It can be shown that the precise gap between the left and right hand sides of this inequality equals the KLD (KL-Divergence) between $q_{\phi}(\z|\x)$ and $p(\z|\x)$, i.e.,
\begin{equation}
    \log \, p_{\psi}(\x) = ELBO + KL[\,q_{\phi}(\z|\x) \lVert p_{\phi}(\z|\x)\,].
\end{equation}

VAEs then view the negative ELBO as a loss and, rather than maximising the intractable log-likelihood, aim to minimize
\begin{equation}
\label{eqn:elbo}
    -ELBO(\phi, \psi) = KL[\,q_{\phi}(\z|\x) \lVert p_{\psi}(\z)\,] - \E_{q_{\phi}(\z|\x)}[\,\log (p_{\psi}(\x|\z))\,],
\end{equation}
where the first term is known as the KLD loss and the second term the reconstruction loss. In principle, the prior,  posterior, and generator can be arbitrary distributions. In practice, however, this typically leads to an intractable ELBO. Thus, we assume they are from the family of Gaussian distributions with diagonal covariance matrices, and that the prior is  the isotropic unit Gaussian. There is no real loss of generality in light of the universal approximation theorem. Specifically, we assume $q_{\phi}(\z|\x) \sim \mathcal{N}(\boldsymbol{\mu_{\phi}(\x)}, \boldsymbol{\sigma_{\phi}}(\x))$, $p_\psi(\x|\z)\sim\mathcal{N}(\boldsymbol{\mu_{\psi}(\x)}, \boldsymbol{\sigma_{\psi}}(\x))$, and $p(\z) \sim \mathcal{N}(\boldsymbol{0}, \boldsymbol{I})$, where $\boldsymbol{\sigma_{\phi}}(\x)$ and $\boldsymbol{\sigma_{\psi}}(\x)$ are diagonal matrices.
%\seb{Let's be more specific here on the structure -- write out in addition that $p_\psi(x|Z)\sim\mathcal{N}(.,.)$, $p_\psi(z)\sim\mathcal{N}(\boldsymbol{0};\boldsymbol{I})$, etc..}

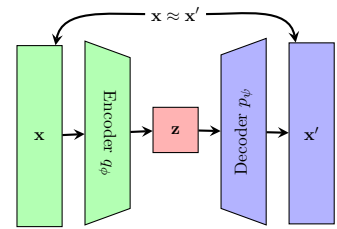
\begin{wrapfigure}[]{r}{0.35\textwidth}
 \vspace*{-1em}
    \centering
    \begin{tikzpicture}[scale=0.6,every node/.style={transform shape},node distance=1cm]
\node (input) [input] {$\mathbf{x}$};
\node (encoder) [encoder, right= of input, xshift=-5em] {Encoder $q_{\phi}$};
\node (latent) [latent, right= of encoder, yshift=5em] {$\mathbf{z}$};
\node (decoder) [decoder, right= of latent, xshift= -5em] {Decoder $p_{\psi}$};
\node (output) [output, right= of decoder, yshift=-5em] {$\mathbf{x'}$};
\node (label) [label, above = of latent, yshift=2em]{$\mathbf{x} \approx \mathbf{x'}$};

\draw [arrow] (input) -- (encoder);
\draw [arrow] (encoder) -- (latent);
\draw [arrow] (latent) -- (decoder);
\draw [arrow] (decoder) -- (output);
% \draw [arrow] (label) .. controls +(left:12cm) and +(left:20mm) ..  (input);
% \draw [arrow] (label) .. controls +(right:12cm) and +(right:20mm) .. (output);
\draw [arrow] (label) to [out=0,in=100] (output);
\draw [arrow] (label) to [out=180,in=80] (input);
\end{tikzpicture}
    \caption{VAE architecture.}
    \label{fig:vae}
    \vspace*{-1em}
\end{wrapfigure}
The neural network  that parametrises $q_{\phi}(\z|\x)$ is called the encoder as it ``encodes'' the data into its latent representation, while the neural net that parameterises $p_{\psi}(\x|\z)$ is called the decoder as it ``decodes'' latent representation to recover the original data.  Figure \ref{fig:vae} shows the typical structure of a VAE. We train the encoding $\phi$  and decoding $\psi$ networks simultaneously but only use the decoder at inference time \new{once a sample of latent features have been chosen. As detailed in Section \ref{sec:latent_sampling}, however, the latent sampling procedure makes use of the encoder.}

\subsection{$\beta$-VAE}

The $\beta$-VAE \cite{higgins2016beta} is a modification of the traditional VAE objective that introduces an adjustable hyperparameter $\beta>0$, precisely the ELBO is modified to
\begin{equation}
    \argmax_{\phi, \psi} \: \E_{q_{\phi}(\z|\x)}[\log (p(\x|\z))] - \beta \, KL(q_{\phi}(\z|\x) \lVert p(\z)).
\end{equation}
Larger values of $\beta$ result in more disentangled latent representations $\z$, while smaller values result in more faithful reconstructions. $\beta$ is often chosen to be greater than one to encourage disentanglement; however, this constrains latent information $\z$ and can lead to poorer reconstructions \cite{higgins2016beta}. Instead, we may choose smaller values of $\beta$ to improve reconstructions and rely on \new{directly sampling from the posterior} to accurately sample from the less structured latent space.

\subsection{Latent Sampling}
\label{sec:latent_sampling}

Typically, samples from a VAE are generated by sampling from the latent prior $p(\z)$ and decoding the result. From a Bayesian perspective, however, conditional on a set of observed data $\boldsymbol{X}$, it is more appropriate to sample from the posterior $p_{\psi}(\z|\boldsymbol{X}) \approx q_{\phi}(\z|\boldsymbol{X}) = \int q_{\phi}(\z|\x)p(\x)d\x$. While this integral is typically intractable, it is possible to sample from. \new{This can be done by first uniformly sampling from the data $\x_0 \sim \boldsymbol{X}$, encoding using the approximate posterior $q_{\phi}(\x_0)$ to obtain the  Gaussian parameters $\mu_{\phi}(\x_0)$ and $\sigma_{\phi}(\x_0)$, and finally sampling latent states from said Gaussian $\z_0 \sim \mathcal{N} (\mu_{\phi}(\x_0), \sigma_{\phi}(\x_0))$. This latent sample $\z_0$ can then be decoded to obtain model parameters which may then be used to construct the IV surface.}

%For this purpose, we use a multivariate Gaussian KDE with bandwidth determined by Scott's Rule \cite{scott2015multivariate} to provide a non-parametric approximation of it.

% Suppose we are given a random sample $\mathbf{Z_1}, \dots, \mathbf{Z_k} \in \R^d$ from a distribution $F$ with density $f$. The KDE estimate of $f$ is defined as $\hat{f}_{KDE}(\z) = \frac{1}{k} \sum_{i=1}^k \zeta_{b}(\z, \mathbf{Z_i})$
% where $\zeta_{b}$ is a kernel function. % (zero centered and variance $1$).
% We take the kernel to be a multivariate Gaussian kernel with bandwidth $b$ and
% use Scott's Rule \cite{scott2015multivariate} to determine the bandwidth. %, so that $b = k^{-1/(d+4)}$.

%\seb{why are $\{\boldsymbol{\mu_{\psi,n}}, \boldsymbol{\sigma_{\psi,n}}\}_{n=1\dots N}$ called the posterior parameters when they parameterise the generative model? I don't see how $q_\phi(z|x)$ gives these to you. reword...}
% Returning to our specific IV context, let $\{\mathbf{x}\}_{n\in\N}$, $\N:=\{1,\dots,N\}$, be the set of historical inputs to the VAE, i.e.,  the SDE model parameters obtained by fitting our historical data. The variational approximation $q_\phi(\z|\mathbf{X})$ provides us with the posterior mean and standard deviations $\{\boldsymbol{\mu_{\phi}(\mathbf{x})},\boldsymbol{\sigma_{\phi}(\mathbf{x})}\}_{n\in\N}$ of the latent factors. For each $\mathbf{x}$, we obtain $M$ samples from the corresponding posterior to obtain a collection of $M \times N$ samples of the latent space, and use these points to approximate $p_{\psi}(\z|\boldsymbol{X})$ via a KDE.

\new{
\subsection{Conditional Variational Autoencoder (CVAE)}
\label{sec:CVAE}
A natural extension of the variational modeling framework is the inclusion of observable market features, such as indices or spot rates. We may then generate surfaces conditional on the state of these features through   Conditional Variational Autoencoders (CVAEs) \cite{sohn2015learning}. Such an approach may be used in risk calculations in which scenarios for the conditional features are generated through other means, and our approach used to generate IV surfaces conditioned on those simulated features.
}

\new{In brief, a CVAE is constructed as follows.
Given ground truth latent factors $\z$, observations $\x$, and conditional features $\y$, we define the generative model $p_{\psi}(\x | \y)p(\y)$ where $p(\y)$ is the prior on the conditional features. Following the same argument as in section \ref{sec:vae} we approximate the posterior $p_{\psi}(\z|\x, \y)$ by a neural network $q_{\phi}(\z|\x, \y)$ and minimize the conditional negative ELBO:
\begin{equation}
    -ELBO(\phi, \psi) = KL[\,q_{\phi}(\z|\x,\y) \lVert p_{\psi}(\z|\y)\,] - \E_{q_{\phi}(\z|\x,\y)}[\,\log (p_{\psi}(\x|\z,\y))\,].
\end{equation}
This allows the conditioning feature $\y$ to modulate  how data $\x$ gets encoded and how a latent factor $\z$ gets decoded. In the results section, we include the VIX as a conditioning feature and find that it improves the out of sample performance of the VAE model. In implementation, the CVAE model is identical to the VAE model, except that the conditional features are added into the encoding and decoding networks.
}

\section{Results}
\label{sec:results}

\subsection{Market Data}
\label{market}

We apply our hybrid method to IV data for three currency pairs provided by Exchange Data International: AUD-USD, EUR-USD, and CAD-USD  for the 1,900 days between September 18th, 2012 to December 30, 2019. \new{The data is divided equally into the training set (September 18th, 2012 to May 09, 2016) and the testing set (May 10, 2016 to December 30, 2019.)} The data includes option prices with five strikes at each of eight different maturities (1M, 2M, 3M, 6M, 9M, 1Y, 3Y, 5Y).  Foreign exchange option prices are quoted in terms of deltas rather than strikes and in terms of  at-the-money call, risk reversal, and butterfly spread options, rather than simple calls. We use standard formulas \cite{reiswich2012fx} to convert raw quotes into IVs at deltas of 0.1, 0.25, 0.5, 0.75, and 0.9 for each  maturity.
%\maxime{Is the data actually split equally?}
%\brian{Yes, 950 training, 950 testing}

% Due to the nature of foreign exchange options data, which are often quoted with respect to at-the-money call, risk reversal, and butterfly spread options, we will work with normalised strikes $K$ instead of the strike price.

\subsection{SDE Model Specifics}
\label{sec:SDESpec}
%\maxime{Ivan's first comment was ``why 0.3''.}

The CTMC model assumes three regimes with a Wasserstein's distance penalty\footnote{The Wasserstein penalty for the various models are chosen from a grid search and balancing goodness of fit to IV smiles with goodness of fit to the implied risk-neutral densities.} of $0.3$. To reduce the number of parameters to fit, initial regime probabilities $\boldsymbol{\pi}$ are set to $\frac13$. To take advantage of the label invariance of the transition matrix, %an additional constraint was introduced to impose an ordering of the regimes. In particular
the mean $\mu$ of each regime is assumed to be in ascending order, i.e., $\mu_1^n \leq \mu_2^n \leq \mu_3^n$. \new{The structure of the transition matrix is detailed in Appendix \ref{sec:amat}. We fit the CTMC model iteratively by first optimizing for the parameters of the first maturity, then iteratively optimizing the parameters of the $i$'th maturity by holding the parameters of the first $(i -1)$ maturities fixed.}
For the L\'evy additive processes, we fit the parameters subject to  a Wasserstein's distance penalty of $0.1$ for time to maturity (TTM) less than 1 year and $0.3$ otherwise to increase the regularising power at larger TTM.   Moreover, we apply a penalty of $10^{-8}$ on day-to-day parameter percentage changes to stabilise model parameters without sacrificing fit quality. For the Gaussian Mixture JD model,  we assume a mixture of two Gaussians, as adding more factors did not increase the fit quality. The penalty varies with maturity as long maturities tend to have stable pdfs but less stable IVs, while shorter maturities tend to have less stable pdfs. Table \ref{tab:rmse} show the median rmse
across days, where the rmse on a given day is computed across all Delta/maturity pairs, for the collection of models and FX pairs we study.
%\begin{wraptable}{r}{0.6\textwidth}
\begin{table}
\renewcommand{\arraystretch}{1}
\footnotesize
\centering
    \begin{tabular}{rrrr}
    \addlinespace
    \toprule
          & AUD-USD & EUR-USD & CAD-USD \\
    \midrule
    CTMC  &          8.1  &          5.0  &          6.0  \\ 
    DE JD    &        62.0  &        42.9  &        64.2  \\
    GM JD   &        10.2  &        17.0  &        24.3  \\
    CGMY/KoBoL  &      140.2  &      151.3  &      136.2  \\
    \bottomrule
    \end{tabular}
\vspace*{-0.5em}    
\caption{Median rmse ($\times10^{-5}$) for the collection of models and FX pairs.}
\vspace*{-1em}
  \label{tab:rmse}
\end{table}
%\end{wraptable}
%\maxime{A sentence explaining why the penalty is applied in these two flavours at these two locations would strengthen the results.}
%\seb{Long maturity tend to have stable pdfs but price less, while shorter maturities tend to less stability on the pdfs. }

\subsection{VAE Model Specifics}

% \seb{check to see if  this is detailed enough?}
The encoder and decoder of the VAE have four fully connected hidden layers, with 64, 128, 256, and 512 nodes each, and a single output layer mapping to the appropriate dimensions. \new{The network structure was selected using a validation set, however, we found that any network exceeding four layers with a minimum of 64 node in each layer is sufficient to produce satisfactory results. We used ADAM with weight decay (AdamW) with a fixed learning rate of 0.001. Appropriate transformations (normalizations, log-transforms) are performed to ensure standardized inputs. Details of the transformations used can be found in Table \ref{tab:transforms} in Appendix \ref{sec:additiona_tables_figures}.
%\seb{we need to specify which variables are transformed in what manner, as well what are the activations function used... perhaps in the appendix}
} We perform a grid search over $\beta$ values and number of latent dimensions summarized in Table \ref{tab:AUD_res}. The table reports an evaluation metric described in the next subsection. Training is carried out with batches of $200$ randomly sampled days from the training set. We set a fixed training duration of $2,000$ epochs as we find that is usually sufficient to train the $\beta$-VAE.%, however the Gaussian Mixture JD model requires upwards of 4,000 epochs for losses to stabilize. (I might need to check if using the pre-normalised will reduce this to 2000)

\new{
\subsection{Benchmarks}
\label{sec:benchmarks}
We introduce three benchmarks to assess the performance of our approach. These benchmark models all generate distributions of IV surfaces (using only the training data) that we use to assess how close they are to the testing data.  Details on the benchmarks themselves will be given below while the metric we use is described in the next subsection.
}

\new{
The first is a $\beta$-VAE that is fit directly to the set of IVs on the fixed grid of delta and time to maturity (as defined in Section \ref{sec:eval_metric}) without the addition of any arbitrage constraints. This technique is inspired by \cite{bergeron2021autoencoders} (see also \cite{zheng2019gated}), although they favor a flexible point-based method where the inputs to the VAE are arbitrary strikes and times to maturity.  Typically, these point-based approaches are complemented by either penalizing deviations from (static) arbitrage constraints during training or smoothing the resulting surfaces after generation. However, \cite{bergeron2021autoencoders} shows that arbitrage constraints do not enhance the fit and excluding them introduces only negligible amounts of static arbitrage. In our context, the grid-based method is more natural as our data is already structured in this fashion. For comparison purposes, we present the result from a grid of latent dimensions and $\beta$'s consistent with those used for our other approaches. We label this approach as VAE-IV.
%\seb{But let's say that Maxime find that no-arbitrage constraints do not add much to the fit and don't introduce static arbitrage}
}

\new{
As a second benchmark model, we perform a PCA on the trained CTMC model parameters and sample from the dimensionally reduced latent submanifold using a kernel density estimator (KDE). Several choices for the number of latent dimensions are explored in Table \ref{tab:AUD_res}. 
%\seb{how many dimensions are selected, and what criteria are sued to select those dimensions?} 
We choose a Gaussian kernel with a bandwidth selected through 20-fold cross-validation. Gaussian kernels are generally quite flexible, but other choices
(such as Epanechnikov, Triweight, and Triangular) are possible. For discussions on kernel and bandwidth selection more generally see, e.g., \cite{gramacki2018nonparametric}. This PCA approach serves as a simplification of our CTMC-VAE model where the VAE sampler is replaced with a dimensionality reduction technique combined with a KDE sampler.
}

\new{
As a final benchmark, we use the empirical distribution of the training data.
% \seb{do we do a KDE on this as well?}
}

\subsection{Evaluation Metric}
\label{sec:eval_metric}

%\seb{I am not certain of the description in the last part of this paragraph -- seems to not be correct.}

Our goal is to generate arbitrage-free IV surfaces that are faithful to the historical dataset. Here, we describe a natural metric that allows us to assess how well we meet this goal. Let $\mathrm{G}$ and $\mathrm{F}$  denote the probability distribution over IV surfaces for the trained model using the our algorithm and the true distribution, respectively. As Wasserstein distances provide a natural metric on the space of probability measures, we use the $1$-Wasserstein distance between the trained and true distribution as our performance metric. While the true distribution $F$ is unknown, the data provides a finite sample from it at a set of discrete 2-dimensional grid points $\Z:=\{z_i\}_{i\in\G}$ (the collection of Delta/TTM pairs which are observed). \new{Specifically, we look at the collection: $\Z = \D \times \M$ where $\D := \{0.1, 0.25, 0.5 0.75, 0.9\}$ and $\M:= \{1M, 2M, 3M, 6M, 9M, 1Y, 3Y, 5Y \}$.} Further, the trained model's distribution may be estimated by sampling from the posterior distribution in latent space (using the method described in Section \ref{sec:latent_sampling}) and decoding to produce SDE model parameters, which can be mapped to a sample of IVs at the set of grid points $\Z$. The $1$-Wasserstein distance between the true and the model's distribution may be estimated by the $1$-Wasserstein distance between the multi-variate distribution of IVs at grid points $\Z$ for the test data and the model generated ones. We refer to this quantity as the Wasserstein metric.

\subsection{Results Summary}

%\maxime{Suggestion from Ivan: highlight the best numbers in the table so they are easy to spot. Mention take-away in the caption.}

\begin{table}[htbp]
\tiny
\renewcommand{\arraystretch}{0.95}

  \centering
    \begin{tabular}{rcrcrrrrrrrrrrrrrr}
    \addlinespace
    \toprule
          &       &       &       & \multicolumn{ 14}{c}{Latent Dimension}                                                                       \\
          
    \cmidrule(lr){4-18}
    \multicolumn{3}{c}{Model}              &       & \multicolumn{ 4}{c}{AUD-USD}  &       & \multicolumn{ 4}{c}{EUR-USD}  &       & \multicolumn{ 4}{c}{CAD-USD}  
    \\
    \cmidrule(lr){1-3}
    \cmidrule(lr){5-8}
    \cmidrule(lr){10-13}
    \cmidrule(lr){15-18}
    & & $\beta$      &       & 3     & 5     & 10    & 15    &       & 3     & 5     & 10    & 15    &       & 3     & 5     & 10    & 15 
          \\
    \cmidrule(lr){3-0}
    \cmidrule(lr){5-8}
    \cmidrule(lr){10-13}
    \cmidrule(lr){15-18}
    \parbox[t]{2mm}{\multirow{20}{*}{\rotatebox[origin=c]{90}{VAE}}}   & \parbox[t]{2mm}{\multirow{4}{*}{\rotatebox[origin=c]{90}{CTMC}}}  
                  & 0.01  &       &        4.56  &        5.35  &        4.82  &        4.40  &       &        3.88  &        3.63  &        3.61  &        3.97  &       &        1.54  &        2.12  &        1.78  &        1.63  \\
          &       & 0.1   &       &        4.32  &        5.83  &        4.54  &        4.62  &       &        3.18  &        3.29  &        3.62  &        \bf{2.77}  &       &        1.81  &        1.53  &        1.40  &        \bf{1.34}  \\
          &       & 1     &       &        4.67  &        4.43  &        4.13  &        \bf{3.69}  &       &        3.64  &        3.94  &        3.12  &        3.62  &       &        1.72  &        1.55  &        1.48  &        1.70  \\
          &       & 10    &       &        6.10  &        5.92  &        6.40  &        5.63  &       &        4.00  &        3.59  &        3.86  &        3.90  &       &        2.98  &        2.96  &        3.17  &        3.06  \\
\cmidrule(lr){2-18}
          & \parbox[t]{2mm}{\multirow{4}{*}{\rotatebox[origin=c]{90}{DE}}}    
                  & 0.01  &       &        5.13  &        5.37  &        5.30  &        4.66  &       &        3.52  &        3.79  &        3.88  &        3.69  &       &        1.61  &        1.76  &        1.96  &        2.09  \\
          &       & 0.1   &       &        4.61  &        5.29  &        5.04  &        \bf{4.41}  &       &        4.58  &        3.96  &        3.69  &        3.87  &       &        1.51  &        1.95  &        1.50  &        \bf{1.47}  \\
          &       & 1     &       &        5.84  &        4.92  &        4.61  &        4.81  &       &        3.70  &        3.46  &        4.12  &        4.01  &       &        1.65  &        1.76  &        1.94  &        1.51  \\
          &       & 10    &       &        6.01  &        5.25  &        5.40  &        5.08  &       &        \bf{3.38}  &        3.73  &        3.89  &        3.69  &       &        1.84  &        2.28  &        1.82  &        2.20  \\
\cmidrule(lr){2-18}
          & \parbox[t]{2mm}{\multirow{4}{*}{\rotatebox[origin=c]{90}{GM}}}    
                  & 0.01  &       & 5.49  & 5.14  & 5.61  & 5.48  &       & 3.65  & 3.67  & 3.56  & 4.55  &       & 1.86  & 1.73  & 1.88  & 1.63 \\
          &       & 0.1   &       & 5.94  & 5.04  & 5.02  & 5.13  &       & 3.49  & 3.36  & 3.59  & 3.54  &       & \bf{1.54}  & 1.70  & \bf{1.54}  & 1.91 \\
          &       & 1     &       & 5.15  & 5.27  & 5.00  & 5.83  &       & 4.01  & 3.49  & 3.96  & 3.30  &       & 2.11  & 2.28  & 1.93  & 1.62 \\
          &       & 10    &       & 5.72  & 5.56  & \bf{4.83}  & 5.26  &       & 3.22  & \bf{3.15}  & 3.73  & 3.27  &       & 1.81  & 1.70  & 2.03  & 2.11 \\
\cmidrule(lr){2-18}
          & \parbox[t]{2mm}{\multirow{4}{*}{\rotatebox[origin=c]{90}{$\stackrel{\text{CGMY}/}{\text{KoBoL}}$}}} 
                  & 0.01  &       & 5.86  & 6.11  & 5.80   & \bf{5.53}  &       & 3.58  & \bf{3.36}  & 3.94  & 4.02  &       & 2.33  & 2.08  & \bf{1.85}  & 1.94 \\
          &       & 0.1   &       & 6.12  & 6.37  & 5.99  & 5.65  &       & 3.68  & 4.17  & 3.92  & 3.56  &       & 2.19  & \bf{1.85}  & 2.06  & 2.08 \\
          &       & 1     &       & 6.54  & 6.43  & 6.17  & 6.22  &       & 3.58  & 3.79  & 3.94  & 3.88  &       & 2.09  & 2.34  & 1.97  & 2.28 \\
          &       & 10    &       & 5.96  & 5.97  & 5.90   & 6.20   &       & 4.01  & 3.84  & 4.31  & 4.21  &       & 2.16  & 2.24  & 2.21  & 2.61 \\
\cmidrule(lr){2-18}
          & \parbox[t]{2mm}{\multirow{4}{*}{\rotatebox[origin=c]{90}{IV}}}    
                  & 0.01  &       & 6.07  & 6.16  & 6.06  & 6.26  &       & 3.55  & 4.17  & 4.05  & 3.65  &       & 1.65  & 1.61  & 1.62  & 1.64 \\
          &       & 0.1   &       & 6.22  & 6.34  & 5.87  & 5.86  &       & 3.94  & 3.61  & \bf{3.44}  & 4.04  &       & 1.68  & \bf{1.44}  & 1.59  & 1.86 \\
          &       & 1     &       & 6.43  & 6.55  & 6.05  & \bf{5.60}   &       & 3.88  & 4.14  & 4.11  & 4.08  &       & 2.05  & 1.97  & 1.83  & 1.79 \\
          &       & 10    &       & 6.23  & 6.21  & 6.16  & 6.11  &       & 4.37  & 4.14  & 3.88  & 4.28  &       & 1.82  & 1.76  & 2.32  & 1.75 \\
\cmidrule(lr){1-18} 
        \multicolumn{3}{c}{PCA}        &       & \bf{5.25}  & 5.34  & 7.15  & 8.38  &       & \bf{2.51}  & 3.00     & 3.5   & 6.15  &       & 2.22  & \bf{1.73}  & 2.40   & 2.76 \\
\cmidrule(lr){1-18}         
         \multicolumn{3}{c}{Empirical}        &       & 5.82  & 5.82  & 5.82  & 5.82  &       & 3.83  & 3.83  & 3.83  & 3.83  &       & 1.67  & 1.67  & 1.67  & 1.67 \\
    \bottomrule
    \end{tabular}
\caption{Wasserstein metrics ($\times 10^{-2}$)  for varying levels of $\beta$, latent dimensionality, and currency pair. \new{Three benchmarks are shown here for reference: the Wasserstein metric ($\times 10^{-2}$) between the test set and the VAE-IV, PCA, and training empirical models as described in Section \ref{sec:benchmarks}. Bold numbers are the smallest metric within each subgrid.} \label{tab:AUD_res}
}
\end{table}
Table \ref{tab:AUD_res} shows a complete summary of the Wasserstein metric computed for each  currency pair on a range of $\beta$ values and latent dimensions. Increasing the number of  latent dimensions does not generally increase performance. This suggests that for these  currency pairs, most surfaces can be captured with \new{as few as three factors when viewed holistically. However, it does appear that the optimal hyperparameter pair are often at somewhat higher latent dimensions ($\geq$10) which is natural when there are no penalizations on dimensionality.} Interestingly, for both classes of SDE models, decreasing $\beta$ does not lead to significantly worse performance. This is an indication that the posterior sampling from Section \ref{sec:latent_sampling} performed admirably in highly unstructured latent spaces, which is a by-product of low $\beta$ values. For the three L\'evy additive processes explored here, \new{the double exponential model performs the best but it is generally worse performing than the CTMC model. Moreover, the CTMC is able to significantly outperform most of the benchmark methods. }

%It is interesting to note that the L\'evy additive processes generally perform as well as the others for the majority of the currency pairs explored here despite the fact that they are restricted by the additional constraint of risk-neutrality.} \seb{this sentence needs work.}

\new{
Overall, the results show that our generated surfaces are as close to the testing data's distribution as the training set itself! This suggests that to improve our model's performance we need to include additional explanatory features (Section \ref{sec:cvae_result}) or perhaps a temporal structure.}

\begin{figure}
 
\centering
\vspace*{-0.5em}
\begin{subfigure}{0.8\textwidth}
  \centering
  \includegraphics[width=\textwidth]{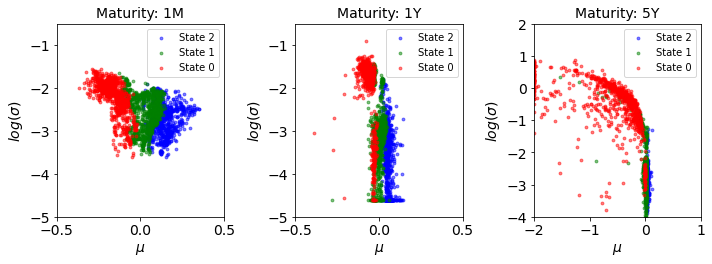}
  \captionof{figure}{ Testing Parameters}
\end{subfigure}%
\vspace*{-0.5em}
\begin{subfigure}{0.8\textwidth}
  \centering
  \includegraphics[width=\textwidth]{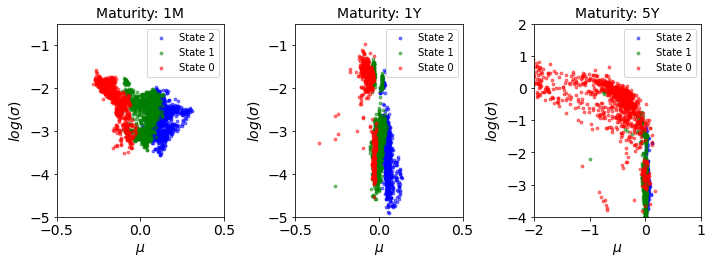}
  \captionof{figure}{ Generated Parameters}
\end{subfigure}
\vspace*{-0.5em}
% \begin{subfigure}{0.6\textwidth}
%   \centering
%   \includegraphics[width=\textwidth]{Figures/AUD_samples2.jpg}
%   \captionof{figure}{Generated Samples} \label{fig:AUD_samples}
% \end{subfigure}
\vspace*{-0.5em}
\caption{\new{Scatter plots of the three regime specific $\mu$ and $\sigma$ from (a) fitting to test data, and (b) random sample from the generative model, using the CTMC model on AUD-USD data.  Colours represent three different regimes.}}
\label{fig:space_scatters}
%\vspace*{-1em}
\end{figure}
Figure \ref{fig:space_scatters} shows a comparison between (a) the CTMC parameters obtained by fitting to test data, and (b) random samples from the corresponding VAE model. \new{We focus on the two most important parameters \textbf{$\mu$} and \textbf{$\sigma$} (which are state and maturity specific) of the CTMC model. We select three maturities 1M, 1Y, and 5Y to succinctly illustrate the results. Ideally, this comparison would be made using the IV surfaces directly implied by the the parameters; however, no good visualization is available for such comparisons. Instead, we illustrate the similarity of the generated and test data distributions using the model parameters. The figure showcases the VAE's ability to capture the complex structures of the CTMC model parameters which in turn is used to generate the final implied volatility surfaces.} As the figure shows, the VAE successfully captures the various complex structures that are inherent in the test data across all maturities and states. 

\begin{figure}
\centering
\includegraphics[width=0.6\textwidth]{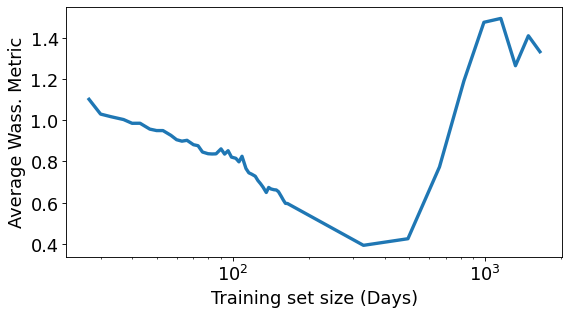}
\vspace*{-1em}
\caption{\new{Average scores of the CTMC model when trained on a range of different sized training sets and tested on the period from January 15th, 2019 to December 30, 2019.}}
\label{fig:training}
\vspace*{-1em}
\end{figure}
\new{Next, we investigate the impact  the training window has on our results. To this end, Figure \ref{fig:training} summarizes how the length of the training set affects the Wasserstein metric between the testing data and the CTMC-VAE model. The score is computed as the average Wasserstein metric, using a randomly generated sample of 500 surfaces and a predefined testing set, across sixteen different sets of hyperparameters as described in Table \ref{tab:AUD_res}. The testing data is fixed to be the period from January 15th, 2019 to December 30th, 2019. 
Here, we reduced the size of the test set in order to illustrate how the size of the training set may both improve and worsen the model's performance.
% Due to the lack of historical data, we reduced the testing set initially defined in Section \ref{sec:results}. 
The ending date of the training data is fixed to be January 14th, 2019, while the starting date ranges from September 18th, 2012 to December 5th, 2018. The figure suggests that approximately 350 days of data is sufficient to fully train the network. The average score increases as the training set size is reduced beyond this point. In contrast, larger time horizon training sets typically produce worse scores as the training set becomes less representative of the current state of the market. This leads to a generative model that may be historically accurate but does not reflect the data in the near future.} 

\begin{figure}
\centering
\includegraphics[width=\textwidth]{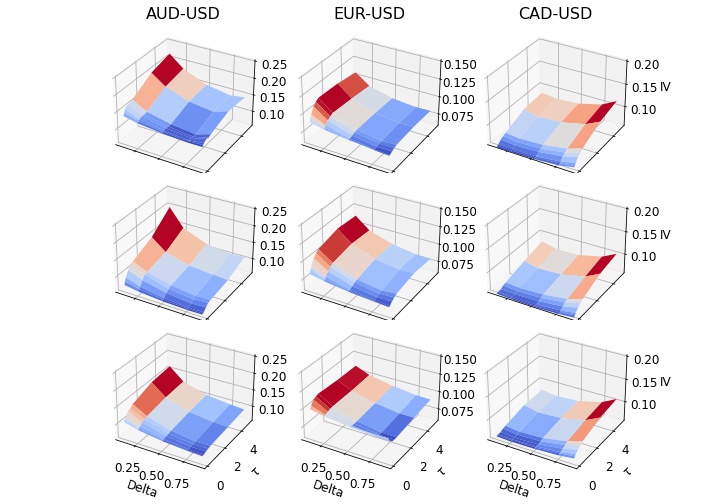}
\vspace*{-0.5em}
\caption{\new{Sample of three randomly generated surfaces using the CTMC-VAE for each of the three currency pairs.}}
\label{fig:samp_surf}
\vspace*{-1em}
\end{figure}
\new{While we show the sampling of model parameters in Figure \ref{fig:space_scatters}, it is informative to  generate surfaces themselves. For this purpose, we show three randomly generated surfaces from each of the three currency pairs in Figure \ref{fig:samp_surf}. There are clear differences between the surfaces for a specific currency pair, however, the general characteristics (skew, level and smile) are similar for a fixed pair. This demonstrates the model's ability to capture the innate characteristics of each currency pair but still faithfully respect the observed variation between random samples.}

\subsection{CVAE Results}
\label{sec:cvae_result}

\new{To test the efficacy of the CVAE approach, we focus on using the daily closing CBOE Volatility Index (VIX) as a predictor for generating IV surfaces in the testing set, courtesy of Wharton Data Services \cite{vixdata}. Specifically, we train the CVAE on several currency pairs conditional on the end-of-day VIX index value. We then condition on the value of the VIX index for each day in the testing set to generate surfaces by randomly sampling from the latent space as specified in section \ref{sec:latent_sampling}. Table \ref{tab:cond_res} in the Appendix details the results while Table \ref{tab:cond_sum} summarizes the comparison between the CVAE, CTMC-VAE, and the benchmarks described in Section \ref{sec:benchmarks}. We train our model using the same training set as  in Section \ref{sec:results}. The results in the table show that the VIX index has significant predictive power on IV surfaces, as the generated surfaces conditional on the VIX index produces significantly smaller metric compared to both the unconditional CTMC-VAE and all benchmark approaches. }
\begin{table}[h]
\renewcommand{\arraystretch}{0.95}

\centering
\footnotesize
\begin{tabular}{rrrrrrrrrrrrr}
 \toprule
%  \multicolumn{2}{c}{Characteristic Functions of Models}
%  \\
%  \midrule
 & \multicolumn{12}{c}{Latent Dimension}\\
\cmidrule(lr){2-13}
 & \multicolumn{4}{c}{AUD-USD} & \multicolumn{4}{c}{EUR-USD} & \multicolumn{4}{c}{CAD-USD} \\
 \cmidrule(lr){2-5}\cmidrule(lr){6-9}\cmidrule(lr){10-13}
 Model  & 3 & 5 & 10 & 15 & 3 & 5 & 10 & 15 & 3 & 5 & 10 & 15\\
 \midrule
 CTMC-CVAE &  3.24 & 3.65 & 3.35 & 3.70 & 2.64 & 2.72 & 2.75 & 2.76 & 1.47 & 1.38 & 1.26 & 1.37\\
 CTMC-VAE  &  4.91 & 5.38 & 4.97 & 4.59 & 3.68 & 3.61 & 3.55 & 3.57 & 2.01 & 2.04 & 1.96 & 1.93\\
 IV-VAE    &  6.24 & 6.32 & 6.04 & 5.96 & 3.94 & 4.01 & 3.87 & 4.02 & 1.80 & 1.69 & 1.84 & 1.76\\
 PCA  &  5.25 & 5.34 & 7.15 & 8.38 & 2.51 & 3.00 & 3.50 & 6.15 & 2.22 & 1.73 & 2.40 & 2.76\\
 Empirical   &  5.82 & 5.82 & 5.82 & 5.82 & 3.83 & 3.83 & 3.83 & 3.83 & 1.67 & 1.67 & 1.67 & 1.67\\
 \toprule
\end{tabular}
\vspace*{-0.5em}    
\caption{\new{Average Wasserstein's metric ($\times 10^{-2}$) across all $\beta$ values for different number of latent dimensions used for the CTMC-CVAE, CTMC-VAE, and benchmarks models IV-VAE, PCA, and Empirical for three currency pairs.}
\label{tab:cond_sum}}
\vspace*{-2em}
\end{table}

\section{Conclusions}

%\maxime{Want the paper to end on a strong note.}

\new{Overall, the results show that our generated surfaces are as close to the testing data's distribution as the training set itself! This suggests that to improve our model's performance we need to include additional explanatory features (Section \ref{sec:cvae_result}) or perhaps a temporal structure.}

To summarise, we propose a hybrid approach for generating synthetic IV surfaces by first calibrating SDE model parameters to historical data -- using a Wasserstein penalty between the implied model risk-neutral distribution and that induced by option data as a regularisation term -- and then training a rich VAE model to learn the distribution on the space of SDE model parameters. We show that the distribution of IV surfaces from the VAE model is \new{capable of generating surfaces as close to the testing data's distribution as the training set itself}, and performs well in comparison with \new{several benchmarks}, while ensuring the generated surfaces are arbitrage-free.
%\maxime{very close should be replaced by the  observation deeper in the text that they are about as close as you could hope to get given the training set that was used}

\new{A short demo of the CTMC model and VAE model fitting procedure are available\footnote{Please note that some notebooks require significant run time due to the adaption of C++ to Python.} at: \url{https://github.com/BrianNingUT/ArbFreeIV-VAE}. }

\bibliographystyle{siamplain}
\bibliography{references}

%\newpage

\appendix

\section{Candidate Risk Neutral Density}
\label{sec:ERDens}
{

Foreign exchange data is often quoted only at specific strikes/deltas. In order to reduce the possibility of overfitting, we introduce a candidate risk neutral density which we attempt to minimize the 1-Wasserstein distance to. We can approximate this candidate density by interpolating implied volatilities at each fixed maturity using B-splines. It is important to note that as this is simply a candidate density derived from a spline interpolation of the IV surface, it provides no guarantee on risk-neutrality. 

Let us define such an interpolated surface by $\sigma(K, \tau)$. The price using this implied volatility may be written in terms of the Black-Scholes price of a call option with spot price $S_0$, strike $K$, maturity $\tau$, and risk-neutral interest rate $r$ as
\begin{align}
\label{eqn:BS_mod}
    C(S_0, K, \tau, r) &= S_0\, \Phi(d^+(K,\tau)) - e^{-{r \tau}}\, K\, \Phi(d^-(K,\tau)),
    \\
    d^\pm(K,\tau) &= \frac{1}{ \sigma(K,\tau) \sqrt{\tau}}\left(\log\left(\tfrac{S_0}{K}\right) + ({r} \pm \tfrac{1}{2} \sigma^2(K,\tau)) \tau\right).
\end{align}

It is well known \cite{breeden1978prices} that, in an arbitrage-free model, $\partial_{KK} C(S_0, K, \tau)$ corresponds to the risk-neutral density of the underlying asset price evaluated at $K$. Thus, we can simply take the second derivative of \ref{eqn:BS_mod} to find the implied density. 

To compute this density, for a fixed $\tau$, we may consider all but $K$ fixed and constant.  As we use a spline representation of the implied volatility we may evaluate the candidate risk-neutral density at any point within the range of available values of $K$. For simplicity, we set $r=0$, and drop the dependence of $d^\pm(K,\tau)$ on $\tau$. Thus,
\begin{align}
\label{eqn:spline_candidate_RN_density}
    \frac{\partial}{\partial ^2 K} C(S_0, K, \tau) 
    % &= S_0\frac{\partial}{\partial ^2 K} \Phi\left(\bar{d}^+(K)\right) - K \; \Phi\left(\bar{d}^-(K)\right)
    % \\
    &= S_0\frac{\partial}{\partial K}\left\{ \underbrace{\phi\left(\bar{d}^+(K)\right) \bar{d}^{+\prime}(K)}_{:=t_1(K)} - \underbrace{\left(\Phi\left(\bar{d}^-(K)\right) + K \; \phi\left(\bar{d}^-(K)\right) \bar{d}^{-\prime}(K)\right)}_{:=t_2(K)}
    \right\}
    % \\
    % &= \frac{\partial}{\partial K} S_0 t_1 - t_2
\end{align}
Focusing on the remaining derivatives, we have
\begin{subequations}
\begin{align}
    \frac{\partial}{\partial K} t_1(K) &= -\bar{d}^+(K) \phi\left(\bar{d}^+(K)\right)\bar{d}^{+\prime}(K)^2 + \phi\left(\bar{d}^+(K)\right) \bar{d}^{+\prime\prime}(K)
    \\
    \frac{\partial}{\partial K} t_2(K)
    &= 2 \, \phi\left(\bar{d}^-(K)\right)\bar{d}^{-\prime}(K)
    - K \phi\left(\bar{d}^-(K)\right)\left[ \bar{d}^-(K) \, \bar{d}^{-\prime}(K)^2 - \bar{d}^{-\prime\prime}(K)\right]
\end{align}
and
\begin{align}
    \bar{d}^{+\prime}(K)
    &= \frac{-1}{K\sigma(K)\sqrt{\tau}} + \frac{\log(K) \, \sigma'(K)}{\sigma^2(K)\sqrt{\tau}} + \frac{1}{2}\sigma'(K)\sqrt{\tau}
    \\
    \bar{d}^{-\prime}(K)
    &= \bar{d}^{+\prime}(K) - \sigma'(K)\sqrt{\tau}
    \\
    \bar{d}^{+\prime\prime}(K)
    &= \frac{\sigma(K) + 2K\sigma'(K)}{(K)^2 \, \sigma^2(K)\sqrt{\tau}}
    + \frac {\log(K) \, \sigma(K) \, \sigma''(K) + 2 \, \log(K) \, \sigma'(K)}{\sigma^3(K)\sqrt{\tau}}
    + \frac{1}{2}\sigma''(K)\sqrt{\tau}
    \\
    \bar{d}^{-\prime\prime}(K)
    &= \bar{d}^{+\prime\prime}(K) - \sigma''(K)\sqrt{\tau}\,.
\end{align}
As we use splines for $\sigma(K)$, putting these computations together with \eqref{eqn:spline_candidate_RN_density} provides us with the candidate risk-neutral density which we use to regularise the implied volatility fits to.
\end{subequations}
}

\new{
\section{Effect of $\alpha$}
\label{sec:alpha_eff}
The parameter $\alpha$ serves as a regularising term to prevent overfitting when the data points in the delta axis are sparse. Figure \ref{fig:alpha_plots} shows its effects. We have chosen a day that is particularly difficult to fit to exhibit the effects of $\alpha$. Large values of $\alpha$ (green) correspond to smoother risk-neutral density curves that only deviate slightly from the non-risk-neutral density implied by the spline interpolation at the cost of significantly poor fits to the IV surfaces. In contrast, smaller values of $\alpha$ (blue) typically correspond to rougher risk-neutral densities that often lead to rougher IV surfaces that over-fit to the small number of IV points available. Often, the middle ground (orange), which correctly balances both accuracy and smoothness, is required. In practice, candidate back-testing can determine the optimal choice. We employ a grid search and balance goodness of fit to IV smiles with goodness of fit to the implied risk-neutral densities. For different models, different optimal $\alpha$ are obtained, details are found in Section \ref{sec:SDESpec}. 
%\seb{can we specify what they are?} 
It is important to note that the spline implied density (grey) is not guaranteed to be risk-neutral and, thus, a perfect fit to such densities is sometimes impossible despite choosing a large value of $\alpha.$
}

\begin{figure}

\centering
\includegraphics[width=\textwidth]{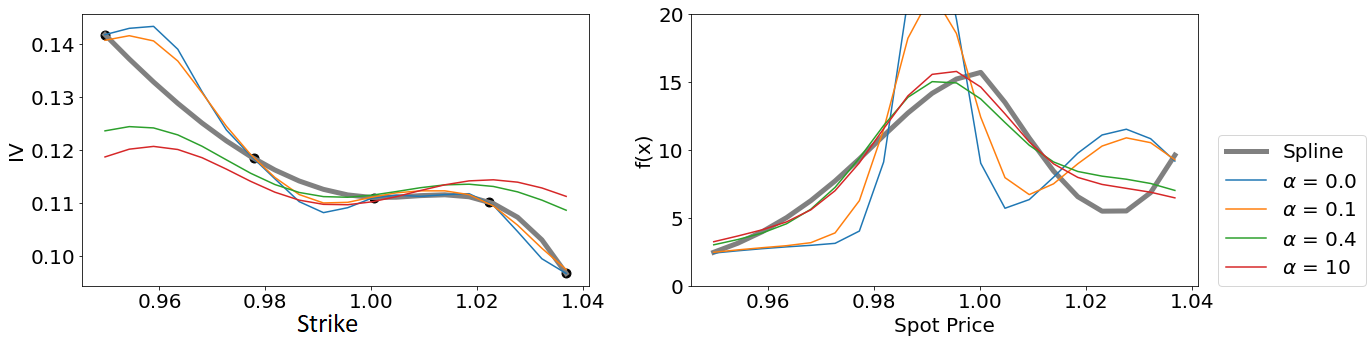}
\vspace*{-1em}
\caption{\new{Fits of the CTMC model to the first maturity (1M) of a particular day (2012-05-16) at varying level of the regularizer $\alpha$. The left panel showcases the fits to the IV surface and the right panel the corresponding risk-neutral density.}
}
\label{fig:alpha_plots}
\vspace*{-1em}
\end{figure}

\section{Proofs}
\begin{proof}[{Proof of Proposition \ref{prop:ctmc}}]

Denote $v_t:=\EQ[e^{iz X_{\tau}}|\mathcal{F}_t]$. As $(X,Z)$ is Markov, there exists a function $v:\K\times\mathds{R}_+\times\mathds{R}\to\mathds{R}$, such that $v_t=v^{Z_t}(t,X_t)$.  Moreover, applying the Feynman-Kac theorem, we have that the function $v^k(t,x)$ satisfies the coupled system of PDEs
\begin{equation}
    \left(\partial_t + (A_{kk}^{(n)} + \mathcal{L}^{k(n)})\right)v^k(x, t) + \sum_{j \neq k} A_{kj}^{(n)} v^j(x,t) = 0,\quad \forall\;t\in[\tau_{n-1},\tau),\; n\in\mathfrak{N},\; k\in\K,
    \label{eqn:PDE-for-v}%
\end{equation}%
subject to the terminal condition (t.c.) $v^k(\tau,x) = e^{iz\,x}$, and $\mathcal{L}^k$ denotes the infinitesimal generator, given $Z_t=k$, which acts upon twice differentiable functions as follows
\begin{equation}
\mathcal{L}^{k(n)} f(x) = \left(\mu_k^n - \tfrac{1}{2} (\sigma_k^n)^2\right) \partial_x f(x)    + \tfrac{1}{2} (\sigma_k^n)^2 \partial_{xx} f(x)    \,.
\end{equation}

To solve the coupled system of PDEs \eqref{eqn:PDE-for-v}, we apply the Fourier transform defined as $\hat{f}(\omega) = \mathcal{F}[f](\omega):= \int_{-\infty}^{\infty} e^{i \omega x}f(x)\,dx$ to both sides of \eqref{eqn:PDE-for-v}. Recall that $
    \mathcal{F}[\partial_x^n f](\omega)
    = i\, \omega \, \mathcal{F}[\partial_x^{n-1}](\omega)
    = \dots
    = (i \omega)^n\, \mathcal{F}[f](\omega)$. Thus,
\begin{equation}
    \mathcal{F}[\mathcal{L}^{k(n)} v^k](\omega, t) = \left(i \left(\mu_k^n - \tfrac{1}{2} (\sigma_k^n)^2\right) \omega - \tfrac{1}{2} (\sigma_k^n)^2 \omega ^2 \right) \mathcal{F}[v^k](\omega, t).
\end{equation}

Using the above, and denoting $\hat{v}^k = \mathcal{F}[v^k]$, \eqref{eqn:PDE-for-v} may be written in Fourier space as
\begin{equation}
    \left[ \partial_t + A^{(n)}_{kk}
    + \gamma^n_k(\omega) \right] \hat{v}^k(t, \omega) + \sum_j A^{(n)}_{kj} \hat{v}^j(t, \omega) = 0\,,
    \quad \forall\;t\in[\tau_{n-1},\tau),\; n\in\mathfrak{N},\; k\in\K,
    \label{eqn:PDE-transformed}
\end{equation}%
s.t. the t.c. $v^k(\tau,x)=\mathcal{D}(z-\omega)$, where $\gamma^n_k(\omega):=i \left(\mu_k^n - \tfrac{1}{2} (\sigma_k^n)^2\right) \omega
    - \tfrac{1}{2} (\sigma_k^n)^2 \omega^2$ and $\mathcal{D}$ is the Dirac delta function.
We may further rewrite this system of equations in matrix notation by (i) defining the matrix $\Psi^n(\omega)$ whose entries are $
    [\Psi(\omega)]_{jk} = \left( A^{(n)}_{kk} +\gamma^n_k(\omega) \right)\delta_{jk} + A^{(n)}_{jk} \;(1 - \delta_{jk})$
where $\delta_{jk}$ is the Kroencker delta, and (ii) defining the vector of transformed prices $\boldsymbol{\hat{v}}(t, \omega)=({\hat{v}}^1(t, \omega),\dots,{\hat{v}}^K(t, \omega))^\intercal$.  Thus, \eqref{eqn:PDE-transformed} may be written as a vector-valued ODE
\begin{align}
    (\partial_t + \Psi(\omega)) \boldsymbol{\hat{v}}(t, \omega) = \boldsymbol{0}\;, \qquad\qquad \quad \forall\;t\in[\tau_{n-1},\tau),\; n\in\mathfrak{N},
\end{align}
s.t. the t.c. $\boldsymbol{\hat{v}}(\tau, \omega) = \mathcal{D}(\omega-z)\,\boldsymbol{1}$.
This system may be solved explicitly by backward induction. For $t \in [\tau_{N-1}, \tau)$ (the last period), the matrix ODE admits the solution
$\boldsymbol{\hat{v}}(t, \omega) = e^{(\tau - t) \Psi(\omega)} \boldsymbol{1}  \;\hat{\varphi}(\omega)\;.
$
Next, due to continuity, we have that
$
     \boldsymbol{\hat{v}}(\tau_{N-1}, \omega) = \lim_{t \downarrow \tau_{N-1}}  \boldsymbol{\hat{v}}(t, \omega) = e^{(\tau - \tau_{N-1}) \Psi(\omega)}\boldsymbol{1}\;\hat{\varphi}(\omega)$.
Using this limit as the t.c. at $t=\tau_{N-1}$, for $t\in(\tau_{N-2},\tau_{N-1}]$, we solve
    $(\partial_t + \Psi_{N-1}(\omega)) \boldsymbol{\hat{v}}(t, \omega) = \boldsymbol{0}$, which admits the solution
\[
\boldsymbol{\hat{v}}(t, \omega) = e^{(\tau_{N-1} - t) \Psi_{N-1}(\omega)} e^{(\tau - \tau_{N-1}) \Psi(\omega)}\boldsymbol{1}  \;  \mathcal{D}(\omega-z).
\]
Continuing iteratively, we obtain
% \begin{equation}
%      \boldsymbol{\hat{v}}(\tau_{N-2}, \omega) = \lim_{t \downarrow \tau_{N-2}}  \boldsymbol{\hat{v}}(t, \omega) = e^{(\tau_{N-1} - \tau_{N-2}) \Psi_{N-1}(\omega)} e^{(\tau - \tau_{N-1}) \Psi(\omega)}\hat{\psi}(\omega)\boldsymbol{1}
% \end{equation}
Continuing iteratively we arrive at
\begin{equation}
    \boldsymbol{\hat{v}}(0, \omega) =  e^{\tau_1 \Psi_{1}(\omega)} e^{(\tau_{2} - \tau_{1}) \Psi_{2}(\omega)} \dots e^{(\tau - \tau_{N-1}) \Psi(\omega)} \boldsymbol{1}\;\mathcal{D}(\omega-z)\;.
\end{equation}
Averaging over the prior $\pi$ on $Z_0$, and taking the Fourier inverse (which is trivial due to the Dirac delta function), we obtain the stated result.
\end{proof}

\section{Structure of the A matrix}
\label{sec:amat}
\new{
We assume each state $k$ has its corresponding rate parameter $\lambda_k$ which determines the rate at which the chain will move out of the state. For any state $1 \leq k \leq K $, the chain has an equal chance of moving into the state below ($k-1$) or above ($k + 1$). We further assume that states form a cyclical graph so that  state $1$ may transition to state $K$, and vice versa. Any other transitions will have probability $0$. This restricts the process to only being able to move through one state at a time without jumping.}

\new{All together, this creates a generator matrix of the form:}
\begin{equation}
    A^n =
    \begin{bmatrix}
    -\lambda_1^n & \frac{\lambda_1^n}{2} & 0 & \dots &\frac{\lambda_1^n}{2}\\
    \frac{\lambda_2^n}{2} & -\lambda_2^n & \frac{\lambda_2^n}{2} & 0 & \dots \\
    \vdots & \ddots\\
    \dots & 0 & \frac{\lambda_{K-1}^n}{2} & -\lambda_{K-1}^n & \frac{\lambda_{K-1}^n}{2} \\
    \frac{\lambda_K^n}{2} & \dots & 0 & \frac{\lambda_K^n}{2} & -\lambda_K^n
    \end{bmatrix}
\end{equation}

\new{which reduces the number of parameters needed to be estimated at each given maturity to $3K$ (K from $\sigma$, K from $\mu$, and K from $\lambda$) excluding the vector of initial probabilities, where $K$ is the number of possible states of the system.}

\section{Additional Tables and Figures}
\label{sec:additiona_tables_figures}

\begin{table}[t!]
\renewcommand{\arraystretch}{1}
\centering

\scriptsize
\begin{tabular}{rrrrrrrrrrrrrr}
 \toprule
%  \multicolumn{2}{c}{Characteristic Functions of Models}
%  \\
%  \midrule
 &  & \multicolumn{12}{c}{Latent Dimension}\\
\cmidrule(lr){3-14}
 &  & \multicolumn{4}{c}{AUD-USD} & \multicolumn{4}{c}{EUR-USD} & \multicolumn{4}{c}{CAD-USD} \\
 \cmidrule(lr){3-6}\cmidrule(lr){7-10}\cmidrule(lr){11-14}
 Model & $\beta$ & 3 & 5 & 10 & 15 & 3 & 5 & 10 & 15 & 3 & 5 & 10 & 15\\
 \midrule
 \multirow{4}{*}{CTMC-CVAE}
 &  0.01 & 3.47 & 4.11 & 3.41 & 3.94 & 3.19 & 3.17 & 3.34 & 3.16 & 1.63 & 1.34 & 1.12 & 1.38\\
 &  0.1  & 3.54 & 3.47 & 3.46 & 4.18 & 2.90 & 3.58 & 3.11 & 3.12 & 1.08 & 1.13 & 1.08 & 1.22\\
 &  1    & 3.99 & 3.90 & 3.27 & 3.49 & 2.48 & 2.23 & 2.55 & 2.81 & 1.18 & 1.06 & 1.13 & 0.91\\
 &  10   & 2.97 & 3.14 & 3.26 & 3.18 & 1.96 & 1.89 & 2.00 & 1.96 & 2.00 & 2.00 & 1.72 & 1.96\\
 
  \midrule
 \multirow{4}{*}{CTMC-VAE}
 &  0.01 & 4.56 & 5.35 & 4.82 & 4.40 & 3.88 & 3.63 & 3.61 & 3.97 & 1.54 & 2.12 & 1.78 & 1.63\\
 &  0.1  & 4.32 & 5.83 & 4.54 & 4.62 & 3.18 & 3.29 & 3.62 & 2.77 & 1.81 & 1.53 & 1.40 & 1.34\\
 &  1    & 4.67 & 4.43 & 4.13 & 3.69 & 3.64 & 3.94 & 3.12 & 3.62 & 1.72 & 1.55 & 1.48 & 1.70\\
 &  10   & 6.10 & 5.92 & 6.40 & 5.63 & 4.00 & 3.59 & 3.86 & 3.90 & 2.98 & 2.96 & 3.17 & 3.06\\
 
  \midrule
 \multirow{4}{*}{IV-VAE(B)}
 &  0.01 & 6.07 & 6.16 & 6.06 & 6.26 & 3.55 & 4.17 & 4.05 & 3.65 & 1.65 & 1.61 & 1.62 & 1.64\\
 &  0.1  & 6.22 & 6.34 & 5.87 & 5.86 & 3.94 & 3.61 & 3.44 & 4.04 & 1.68 & 1.44 & 1.59 & 1.86\\
 &  1    & 6.43 & 6.55 & 6.05 & 5.60 & 3.88 & 4.14 & 4.11 & 4.08 & 2.05 & 1.97 & 1.83 & 1.79\\
 &  10   & 6.23 & 6.21 & 6.16 & 6.11 & 4.37 & 4.14 & 3.88 & 4.28 & 1.82 & 1.76 & 2.32 & 1.75\\
 
   \midrule
CTMC-PCA(B)  & &  5.25 & 5.34 & 7.15 & 8.38 & 2.51 & 3.00 & 3.50 & 6.15 & 2.22 & 1.73 & 2.40 & 2.76\\
 
    \midrule
 Empirical (B)&   &  5.82 & 5.82 & 5.82 & 5.82 & 3.83 & 3.83 & 3.83 & 3.83 & 1.67 & 1.67 & 1.67 & 1.67\\
 
 \toprule
\end{tabular}
\vspace*{-0.5em}
\caption{\new{Wasserstein metrics ($\times 10^{-2}$)  for varying levels of $\beta$, latent dimensionality, and currency pair using the CVAE with the CTMC model with VIX as predictor compared with the vanilla CTMC-VAE and benchmark approaches (B) IV-VAE and CTMC-PCA trained on the period from September 18th, 2012 to May 09, 2016. The Wasserstein metric ($\times 10^{-2}$) between the training and tests sets (Empirical) are presented here for reference.}\label{tab:cond_res}
%\brian{Numbers TBD, see CTMC backup}
~\vspace*{-2em}
}
\end{table}

\begin{table}[h!]
\centering

\small
\begin{threeparttable}
\begin{tabular}{lll}
 \toprule
 Model & Parameter & Transforms\\
 \midrule
 \multirow{4}{*}{CTMC}
 & $\pi_1 \dots \pi_3$ & None \tnote{1} \\
 & $\mu_1 \dots \mu_3$ & $\frac{\mu_i - \mubar(\mu_i)}{s(\mu_i)}$ \\
 &$\sigma_1 \dots \sigma_3$ & $\frac{log(\sigma_i) - \mubar(log(\sigma_i))}{s(log(\sigma_i))}$ \\
 & $\lambda_1, \dots \lambda_3$ & $\frac{log(\lambda_i)- \mubar(log(\lambda_i))}{s(log(\lambda_i))}$ \\
  \midrule
  
  \multirow{4}{*}{DE-JD}
  & $\sigma$ & $\frac{log(\sigma) - \mubar(log(\sigma))} {s(log(\sigma)}$ \\
 & $\lambda$ & $\frac{\lambda_i - \mubar(\lambda_i)}{s(\lambda_i)}$ \\
 &$p$ & $\frac{\tp - \mubar(\tp)}{s(\tp)}, \quad \tp = log \left( \frac{p}{1-p}\right)$ \\
 & $a_1$ & $\frac{a_1 - \mubar(a_1)}{s(a_1)}$ \\
 & $a_2$ & $\frac{a_2 - \mubar(a_2)}{s(a_2)}$ \\
   \midrule
  
  \multirow{4}{*}{GM-JD}
  & $\sigma$ & $\frac{log(\sigma) - \mubar(log(\sigma))} {s(log(\sigma)}$ \\
 & $\lambda$ & $\frac{\lambda_i - \mubar(\lambda_i)}{s(\lambda_i)}$ \\
 &$\tpi_1, \tpi_2$ & $\frac{log(\eta_i) - \mubar(log(\eta_i))}{s(log(\eta_i)}, \quad \eta_1 = 1, \eta_2 = \frac{\tpi_2}{\tpi_1} \dots$ \\
 & $\tmu_1, \tmu_2$ & $\frac{\tmu_i - \mubar(\tmu_i)}{s(\tmu_i)}$ \\
 & $\tsig_1, \tsig_2$ & $\frac{log(\tsig_i) - \mubar(log(\tsig_i))}{s(log(\tsig_i))}$ \\
   \midrule
  
  \multirow{4}{*}{CGMY/KoBoL}
 & $C$ & $\frac{log(C) - \mubar(log(C)))}{s(log(C))}$ \\
 & $G$ & $\frac{log(G) - \mubar(log(G)))}{s(log(G))}$ \\
 & $M$ & $\frac{log(M) - \mubar(log(M)))}{s(log(M))}$ \\
 & $Y$ & $\frac{log(Y) - \mubar(log(Y)))}{s(log(Y))}$ \\

 \toprule
\end{tabular}
\vspace*{-0.5em}
\caption{\new{Normalizations performed on model parameters before input into VAE. $\mubar(\cdot)$ and $s(\cdot)$ are the empirical mean and standard deviation respectively.} \label{tab:transforms}
%\brian{Numbers TBD, see CTMC backup}
~\vspace*{-1em}
}
\begin{tablenotes}
\item[1] $\pi_1=\pi_2=\pi_3=\frac13$
\end{tablenotes}
\end{threeparttable}

\end{table}

\end{document}